\documentclass[twocolumn,10pt]{IEEEtran}
\IEEEoverridecommandlockouts
\usepackage{amsmath}
\usepackage{amsfonts}
\usepackage{amssymb}
\usepackage{amsthm}
\usepackage{graphicx}
\usepackage{psfrag}
\usepackage{cite}
\usepackage{algorithm}
\usepackage{algorithmic}
\usepackage{color}
\usepackage{float}
\usepackage{epsfig,array,multicol,verbatim}
\usepackage{subfigure}
\usepackage{cases}
\usepackage{setspace}
\usepackage{url}

\usepackage[top=0.8in, bottom=0.8in, left=0.8in, right=0.8in]{geometry}

\theoremstyle{plain} \newtheorem{theorem}{Theorem}
\theoremstyle{plain} \newtheorem{proposition}{Proposition}
\theoremstyle{plain} \newtheorem{corollary}{Corollary}
\theoremstyle{remark} \newtheorem{remark}{Remark}
\theoremstyle{remark} \newtheorem{lemma}{Lemma}
\theoremstyle{plain} \newtheorem{definition}{Definition}
\theoremstyle{plain} \newtheorem{assumption}{Assumption}
\theoremstyle{plain} \newtheorem{fact}{Fact}

\begin{document}
\title{Optimal Spectrum Sharing in MIMO Cognitive Radio Networks via Semidefinite Programming}
\author{Ying~Jun~(Angela)~Zhang,~\IEEEmembership{Member,~IEEE} and Anthony~Man-Cho~So
\thanks{This work was supported in part by the Competitive Earmarked
Research Grant (Project Number 418707) established under the
University Grant Committee of Hong Kong, the Direct Grant for Research (Project Numbers 2050401 and 2050439) established by The Chinese University of Hong Kong, and Project \#MMT-p2-09 of the Shun Hing Institute of Advanced Engineering, The Chinese University of Hong Kong.}
\thanks{Y. J. Zhang is with the Department of Information Engineering and the Shun Hing Institute of Advanced Engineering, The Chinese University of Hong Kong, Hong Kong (email: yjzhang@ie.cuhk.edu.hk).}
\thanks{A. M.-C. So is with the Department of Systems Engineering and Engineering Management and the Shun Hing Institute of Advanced Engineering, The Chinese University of Hong Kong, Hong Kong (email: manchoso@se.cuhk.edu.hk)}
}

\maketitle

\begin{abstract}
In cognitive radio (CR) networks with multiple-input multiple-output (MIMO) links, secondary users (SUs) can exploit ``spectrum holes" in the space domain to access the spectrum allocated to a primary system.  However, they need to suppress the interference caused to primary users (PUs), as the secondary system should be transparent to the primary system. In this paper, we study the optimal secondary-link beamforming pattern that balances between the SU's throughput and the interference it causes to PUs. In particular, we aim to maximize the throughput of the SU, while keeping the interference temperature at the primary receivers below a certain threshold.

Unlike traditional MIMO systems, SUs may not have the luxury of knowing the channel state information (CSI) on the links to PUs. This presents a key challenge for a secondary transmitter to steer interference away from primary receivers. In this paper, we consider three scenarios, namely when the secondary transmitter has complete, partial, or no knowledge about the channels to the primary receivers. In particular, when complete CSI is not available, the interference-temperature constraints are to be satisfied with high probability, thus resulting in chance constraints that are typically hard to deal with. Our contribution is fourfold. First, by analyzing the distributional characteristics of MIMO channels, we propose a unified homogeneous quadratically constrained quadratic program (QCQP) formulation that can be applied to all three scenarios, in which different levels of CSI knowledge give rise to either deterministic or probabilistic interference-temperature constraints. The homogeneous QCQP formulation, though non-convex, is amenable to semidefinite programming (SDP) relaxation methods. Secondly, we show that the SDP relaxation admits no gap when the number of primary links is no larger than two. A polynomial-time algorithm is presented to compute the optimal solution to the QCQP problem efficiently. Thirdly, we propose a randomized polynomial-time algorithm for constructing a near-optimal solution to the QCQP problem when there are more than two primary links. Finally, we show that when the secondary transmitter has no CSI on the links to primary receivers, the optimal solution to the QCQP problem can be found by a simple matrix eigenvalue-eigenvector computation, which can be done much more efficiently than solving the QCQP directly.

\end{abstract}


\section{Introduction}
Cognitive radio (CR), which allows secondary users (SUs) to opportunistically utilize the frequency spectrum originally assigned to licensed primary users (PUs), is a promising approach to alleviate spectrum scarcity \cite{H05}. In CR networks with single-antenna nodes, SUs can transmit only when it detects a spectrum hole in either time or frequency domain, so as to avoid causing harmful interference to PUs \cite{P09,TMS09}. Such schemes, however, only work when the primary system severely underutilizes the assigned spectrum. Otherwise, the secondary system would not have adequate chances to access the wireless channel.

Recent development in multiple-input multiple-output (MIMO) antenna techniques opens up a new dimension, namely space, for co-channel users to coexist without causing severe interference to each other \cite{GSS+03}. Indeed, in CR networks where stations are equipped with multiple antennae, SUs can transmit at the same time as the PUs through space-domain signal processing. The nature of CR networks gives rise to several challenging issues that do not exist in traditional MIMO systems. First, SUs are solely responsible for suppressing the interference they cause to PU receivers, as the primary system should not be aware of the existence of the secondary system. That is, we cannot rely on the PUs to do receiver-side interference cancellation.  Secondly, SUs may not have the luxury of knowing the channel state information (CSI) on the links to PUs, as the primary system would not deliberately provide their channel estimation to the secondary system. This imposes difficulty on transmitter-side pre-interference cancellation at SU transmitters. It is therefore necessary to revisit space-domain signal processing in the context of MIMO CR networks. In particular,  SUs need to configure their beamforming patterns in a way that balances between their own throughput and the interference they cause to PUs.

Multi-antenna CR networks were recently studied in \cite{KG08,ZL08,SF09}.
Assuming that CSI on all links is perfectly known to the SUs,  \cite{KG08} formulates the SU beamforming problem as a non-convex optimization problem. A semi-distributed algorithm is proposed to obtain a local optimal solution to the problem. On the other hand, \cite{SF09} assumes that the PU can act as a scheduler for SU transmissions. Under this idealistic assumption, an opportunistic orthogonalization scheme is proposed. In \cite{ZL08}, Zhang and Liang studied capacity-achieving transmit spatial spectrum for a single SU, assuming that the SU has full CSI and there is no interference from PUs to the SU. Insightful solution methods are proposed to provide better intuition that may not be obtainable from a numerical optimization perspective. The issue of imperfect CSI estimation is considered in \cite{ZLC08, GLM09, PVS+09} for multiple-input single-output (MISO) CR networks. Therein, robust optimization problems are formulated to ensure the service qualities for both SUs and PUs are satisfactory in the worst case.

In this paper, we study the problem of optimal secondary-link beamforming. Specifically, we aim to maximize the throughput of the SU under the constraint that the interference to PU receivers is below a certain threshold. In contrast to the previous work \cite{KG08,ZL08,SF09}, we consider three practical scenarios: (1) when the SU transmitter knows both the channel matrices from it to PU receivers and the beamforming patterns at PU receivers; (2) when the SU transmitter does not know the beamformer at PU receivers; (3) when the secondary transmitter knows neither the channel matrices to PU receivers nor the beamforming patterns at PU receivers. Note that the deterministic interference-temperature constraints could be too stringent if the SU does not have full CSI (which is the case in the second and third scenarios), as the SU will need to consider the worst-case channel realization when configuring the beamformer. Fortunately, many wireless applications can tolerate occasional dips in the service quality. In order to have a more efficient utilization of the spectrum, we can take advantage of this opportunity and replace the deterministic constraints by probabilistic interference constraints (also referred to as chance constraints). Chance constraints, however, are typically tougher to deal with than deterministic constraints.
The contribution of this paper is fourfold:
\begin{itemize}
  \item We propose a unified homogeneous quadratically constrained quadratic program (QCQP) formulation that can be applied to all three scenarios mentioned above. In particular, the homogeneous QCQP formulation can accommodate both deterministic and probabilistic interference-temperature constraints.
  \item The QCQP formulation, though non-convex, is amenable to semidefinite programming (SDP) relaxation methods. We show that the SDP relaxation admits no gap with the true optimal solution when the number of PUs is no larger than two. A polynomial-time algorithm is presented to compute the optimal solution efficiently.
  \item When there are more than two PUs, we propose a randomized polynomial-time algorithm that can produce a provably near-optimal solution. Numerical results show that the solution produced by our algorithm almost achieves the optimal value.
  \item In the third scenario where the SU transmitter knows neither the channel matrices nor the beamformer at PU receivers, we show that the optimal solution can be obtained very efficiently through a simple matrix eigenvalue-eigenvector computation. That is, there is no need to solve the QCQP problem in this case.
\end{itemize}

We should emphasize that the incomplete CSI here is not to be confused with that in \cite{ZLC08,GLM09,PVS+09}, where it is assumed that the SU knows the CSI on all links, except that there may be uncertainty in the channel estimation. Robust optimization techniques are employed in these papers to deal with the worst-case channel realization.  The case where only channel statistics is known to the SU transmitter is also studied in \cite{PVS+09}.  This is similar to a special case of the third scenario considered in our paper, when the receivers have only one antenna.

The rest of this paper is organized as follows. The system model is described in Section \ref{sec:2}. In Section \ref{section:3}, we formulate the SU beamforming problem as a series of homogeneous QCQP problems for different scenarios. The SDP relaxations of these homogeneous QCQP problems are then introduced in Section \ref{section:4}.  In Section \ref{sec:exact}, a polynomial-time algorithm for finding an optimal solution to the QCQP problem when the number of primary users is no larger than two is presented. In Section \ref{sec:largeK}, we propose a randomized polynomial-time algorithm to find a near-optimal beamforming solution when there are more than two PUs.  In Section \ref{sec:simulation}, the performance of the proposed schemes is evaluated via simulations. Finally, the paper is concluded in Section \ref{section:conclusions}.

\section{System Model}\label{sec:2}
\subsection{System Setup}
In this paper, we consider a CR network in which a secondary user intends to share the spectrum with a primary system consisting of $K$ primary links.  We shall discuss the possibility of extending the proposed approach to a multiple-secondary-link scenario in Section \ref{section:conclusions}. In the sequel, we use the subscript $S$ to denote the secondary link and the subscript $k$ to denote the $k^{th}$ primary link. Let $M_S$ (or $N_S$) and $M_k$ (or $N_k$) denote the number of transmit (or receive) antennae of the secondary and primary links, respectively.  We use $\mathbf{H}_{S,S} \in \mathbb{C}^{N_S\times M_S}$ to denote the channel matrix from the secondary transmitter to the secondary receiver, and $\mathbf{H}_{k,S} \in \mathbb{C}^{N_k\times M_S}$, $\mathbf{H}_{S,k} \in \mathbb{C}^{N_S\times M_k}$ and $\mathbf{H}_{k,j} \in \mathbb{C}^{N_k\times M_j}$ to denote the channel matrices from the secondary transmitter to the $k^{th}$ primary receiver, from the $k^{th}$ primary transmitter to the secondary receiver, and from the $j^{th}$ primary transmitter to the $k^{th}$ primary receiver, respectively. We assume a Rayleigh fading and rich scattering environment, so that the entries of the channel matrices are independently and identically distributed (i.i.d.) complex Gaussian random variables with zero mean and unit variance. As pointed out in \cite{B03, MZSY08}, in an interference-limited environment, each active link should only transmit one data stream at a time to avoid excessive interference to other links. In this case, we use scalars $x_S$ and $x_k$ to denote the transmitted signals by the secondary transmitter and the $k^{th}$ primary transmitter, respectively. Without loss of generality, suppose that $\mathbb{E}[|x_S|^2]=1$ and $\mathbb{E}[|x_k|^2]=1$.

Let $\mathbf{t}_S$ and $\mathbf{r}_S$ denote the beamforming vectors at the secondary transmitter and receiver, respectively. Likewise, let $\mathbf{t}_k$ and $\mathbf{r}_k$ be the beamforming vectors at the $k^{th}$ primary transmitter and receiver. In particular, we have $\|\mathbf{t}_S\|_2^2=P_S$ and $\|\mathbf{t}_k\|_2^2=P_k$, where $P_S$ and $P_k$ are the transmit power of the secondary and the $k^{th}$ primary links, respectively. Without loss of generality, we normalize the receive beamforming vectors so that $\|\mathbf{r}_S\|_2^2=1$ and $\|\mathbf{r}_k\|_2^2=1$.  Then, the received signal at the secondary receiver after receive beamforming is
$$
y_S=\sqrt{\alpha_{S,S}} \mathbf{r}_S^H\mathbf{H}_{S,S}\mathbf{t}_S x_S +\sum_{k=1}^K \sqrt{\alpha_{S,k}} \mathbf{r}_S^H\mathbf{H}_{S,k}\mathbf{t}_k x_k+\mathbf{r}_S^H \mathbf{n}_S,
$$
where $\alpha_{S,S}$ and $\alpha_{S,k}$ denote the path losses from the secondary transmitter to the secondary receiver and from the $k^{th}$ primary transmitter to the secondary receiver, respectively, and $\mathbf{n}_S\sim\mathcal{CN}(0,N_0 \mathbf{I})$ denotes a circular complex additive Gaussian noise vector at the secondary receiver. As such, the signal to interference and noise ratio (SINR) on the secondary link is given by
\begin{eqnarray} \label{eq:SINR}
   \gamma_S &=& \frac{\alpha_{S,S}\big|\mathbf{r}_S^H\mathbf{H}_{S,S}\mathbf{t}_S\big|^2}{\sum_{k=1}^K\alpha_{S,k} \big|\mathbf{r}_S^H\mathbf{H}_{S,k}\mathbf{t}_k\big|^2 + \|\mathbf{r}_S\|_2^2N_0}\nonumber\\
   &=&\frac{\alpha_{S,S}\big|\mathbf{r}_S^H\mathbf{H}_{S,S}\mathbf{t}_S\big|^2}{\sum_{k=1}^K\alpha_{S,k} \big|\mathbf{r}_S^H\mathbf{H}_{S,k}\mathbf{t}_k\big|^2 + N_0}.
\end{eqnarray}
The secondary link's transmission causes an interference signal $\sqrt{\alpha_{k,S}} \mathbf{r}_k^H \mathbf{H}_{k,S} \mathbf{t}_S x_S$ at the output of the $k^{th}$ primary receiver, resulting in an interference power of $
\alpha_{k,S}\big|\mathbf{r}_k^H \mathbf{H}_{k,S}\mathbf{t}_S\big|^2$,
where $\alpha_{k,S}$ is the path loss from the secondary transmitter to the $k^{th}$ primary receiver.

Before leaving this subsection, we emphasize that channel matrices on different links are independent. Furthermore, the beamforming vectors $\mathbf{t}_k$ and $\mathbf{r}_k$ are solely determined by the channels between the nodes in the primary system, as the secondary system is transparent to the primary system. Therefore, $\mathbf{t}_k$ and $\mathbf{r}_k$ are independent of $\mathbf{H}_{k,S}$, $\mathbf{H}_{S,S}$, and $\mathbf{H}_{S,k}$.

\subsection{Objective and Assumptions}

In this paper, we aim to find for the secondary link the optimal beamforming vectors $\mathbf{t}_S$ and $\mathbf{r}_S$ so that the SINR on the secondary link is maximized, while the interference to primary link $k$ is below a tolerable threshold $\epsilon_k$. Mathematically, this can be formulated as the following optimization problem:
\begin{subequations}\label{eqn:generalform}
\begin{eqnarray}
\max_{\mathbf{t}_S, \mathbf{r}_S}&& \gamma_S \\
\text{s.t. }&&\alpha_{k,S}\big| \mathbf{r}_k^H \mathbf{H}_{k,S}\mathbf{t}_S\big|^2 \leq \epsilon_k \qquad\forall k=1,\ldots,K, \label{constraint:10b}\\
&&\|\mathbf{t}_S\|_2^2\leq P_{S,max},
\end{eqnarray}
\end{subequations}
where $P_{S,max}$ is the maximum transmission power of the secondary link.
Throughout this paper, we assume that the path losses $\alpha_{S,S}$, $\alpha_{S,k}$, $\alpha_{k,S}$, and $\alpha_{k,j}$ do not vary significantly within the time period of interest, and hence are known to all stations. It is also reasonable to assume that the secondary user knows its own channel $\mathbf{H}_{S,S}$ at both transmitter and receiver sides. Moreover, the secondary receiver can estimate $\mathbf{H}_{S,k}\mathbf{t}_k$ for all $k$ by overhearing the transmission of primary transmitters.

In practice, however, the secondary transmitter may not know the CSI on the links between primary receivers, because primary receivers would not purposely provide CSI to the secondary system. In this paper, we are interested in the following three different scenarios:

\noindent \emph{\textbf{Scenario 1:}} The secondary transmitter has perfect knowledge of the vector $\mathbf{H}_{k,S}^H\mathbf{r}_k$.

This scenario corresponds to a time division duplex (TDD) system in which channels are reciprocal and the primary receivers use the same beamforming vectors for both reception and transmission. In this case, the secondary transmitter can estimate $\mathbf{H}_{k,S}^H \mathbf{r}_k$ by overhearing the transmission of primary receivers.

\noindent \emph{\textbf{Scenario 2:}} The secondary transmitter knows $\mathbf{H}_{k,S}$ but not $\mathbf{r}_k$.

This scenario corresponds to a TDD system in which primary receivers do not use the beamforming vector $\mathbf{r}_k$ for transmission.

\noindent \emph{\textbf{Scenario 3:}} The secondary transmitter has no knowledge about $\mathbf{H}_{k,S}$ and $\mathbf{r}_k$.

This scenario corresponds to the case where the secondary link has no way to estimate the channel from the primary receiver.

Note that in both Scenarios 2 and 3, constraints \eqref{constraint:10b} are no longer well defined. Indeed, for any given $\mathbf{t}_S$, the value of the left-hand-side is uncertain to the secondary transmitter, as the realizations of $\mathbf{r}_k$ and/or $\mathbf{H}_{k,S}$ are unknown. Therefore, a revision is necessary. One way is to guarantee that the constraints are always satisfied regardless of the realizations of $\mathbf{r}_k$ and $\mathbf{H}_{k,S}$. Then, \eqref{constraint:10b} can be replaced by
\begin{equation}\label{eqn:worstcase}
\max ~\alpha_{k,S}\big| \mathbf{r}_k^H \mathbf{H}_{k,S}\mathbf{t}_S\big|^2 \leq \epsilon_k \qquad\forall k=1,\ldots,K,
\end{equation}
where the maximization is taken over $\mathbf{r}_k$ for Scenario 2, and over $\mathbf{r}_k$ and $\mathbf{H}_{k,S}$ for Scenario 3.

Besides the worst-case guarantee, in practical applications, we may allow the interference to exceed a certain threshold $\epsilon_k$ with a small outage probability $\delta_k$. In this case, \eqref{constraint:10b} can be replaced by
\begin{equation}\label{eqn:chance}
\Pr\left\{\alpha_{k,S}\big| \mathbf{r}_k^H \mathbf{H}_{k,S}\mathbf{t}_S\big|^2 \leq \epsilon_k\right\}\geq 1-\delta_k \qquad\forall k=1,\ldots,K,
\end{equation}
where the probability is taken over $\mathbf{r}_k$ for Scenario 2, and over $\mathbf{r}_k$ and $\mathbf{H}_{k,S}$ for Scenario 3.  Note that \eqref{eqn:chance} is equivalent to \eqref{eqn:worstcase} when $\delta_k=0$ for all $k$.

\subsection{Distribution of $\mathbf{r}_k$}
In this subsection, we discuss the probability distribution of the beamforming vector $\mathbf{r}_k$ at PU receivers. The results will be useful later when we address the chance constraints in \eqref{eqn:chance}.

As mentioned, $\mathbf{r}_k$ is solely determined by the channels between primary nodes. Consider the matrix
\begin{eqnarray*}
\mathbf{\hat{H}}&=&\big[\sqrt{\alpha_{k,1}}\mathbf{H}_{k,1}\mathbf{t}_1, \ldots, \sqrt{\alpha_{k,k}}\mathbf{H}_{k,k}\mathbf{t}_k ,\nonumber\\
&&~~ \ldots, \sqrt{\alpha_{k,K}}\mathbf{H}_{k,K}\mathbf{t}_K\big] \\
&:=&\big[\mathbf{\hat{h}}_1,  \ldots, \mathbf{\hat{h}}_k, \ldots,\mathbf{\hat{h}}_K\big],
\end{eqnarray*}
where $\alpha_{k,j}$ is the path loss from the $j^{th}$ PU transmitter to the $k^{th}$ PU receiver, and $\mathbf{\hat{h}}_1, \ldots,\mathbf{\hat{h}}_K$ are the columns of $\mathbf{\hat{H}}$, which are independent of each other.  Let $\mathbf{\hat{H}}_{(-k)}$ be the $N_k\times(K-1)$ matrix obtained by deleting the $k^{th}$ column of $\mathbf{\hat{H}}$. Then, in general, the vector $\mathbf{r}_k$ takes the form
\begin{equation}\label{eqn:rk}
\mathbf{r}_k=\beta_k \mathbf{W}_k \mathbf{\hat{h}}_k,
\end{equation}
where $\beta_k=\frac{1}{|| \mathbf{W}_k \mathbf{\hat{h}}_k||_2}$ is a normalization factor that ensures $\|\mathbf{r}_k\|_2^2=1$, and $\mathbf{W}_k$ is a random Hermitian matrix that is independent of $\mathbf{\hat{h}}_k$. In particular, we have $\mathbf{W}_k=\mathbf{I}$ for the matched-filter (MF) receiver,
\begin{equation}\label{eqn:ZF}
\mathbf{W}_k=\mathbf{I}-\mathbf{\hat{H}}_{(-k)}(\mathbf{\hat{H}}_{(-k)}^H\mathbf{\hat{H}}_{(-k)})^{-1}\mathbf{\hat{H}}_{(-k)}^H
\end{equation}
for the zero-forcing (ZF) receiver, and
\begin{equation}\label{eqn:MMSE}
\mathbf{W}_k=(\mathbf{\hat{H}}_{(-k)}\mathbf{\hat{H}}_{(-k)}^H+N_0\mathbf{I})^{-1}
\end{equation}
for the minimum-mean-squared-error (MMSE) receiver. Before proceeding, let us state a definition and introduce two assumptions.
\begin{definition}
A random vector $\mathbf{x}$ is called a normalized complex Gaussian vector if $\mathbf{x}=\frac{\mathbf{z}}{\|\mathbf{z}\|_2}$, where $\mathbf{z}\sim \mathcal{CN}(\mathbf{0}, \mathbf{I})$. In particular, $\|\mathbf{x}\|_2=1$. A normalized complex Gaussian vector is an isotropically distributed unit vector\footnote{A unit vector is said to be isotropically distributed if it is equally likely to point in any direction in the complex space. In other words, the vector is uniformly distributed on a complex unit sphere.}.
\end{definition}

\begin{assumption}\label{ass:1}
The $k^{th}$ column of $\hat{\mathbf{H}}$, i.e., $\mathbf{\hat{h}}_k$, has the same distribution as $\alpha \mathbf{x}$, where $\mathbf{x}$ is a normalized complex Gaussian vector, and $\alpha$ is a scaling factor.
\end{assumption}
Assumption \ref{ass:1} is valid for most practical MIMO systems. Consider the singular value decomposition (SVD) of $\mathbf{H}_{k,k}=\mathbf{U}\mathbf{\Lambda}\mathbf{V}^H$, where $\mathbf{U}$  and $\mathbf{V}$ are unitary matrices containing the left and right singular vectors, respectively, and $\mathbf{\Lambda}$ is a diagonal matrix containing the singular values. It is known that the columns of $\mathbf{U}$ and $\mathbf{V}$ have the same distribution as a normalized complex Gaussian vector \cite{MH99}. In the case of single-user precoding, it is optimal to set $\mathbf{t}_k$ to be proportional to $\mathbf{v}_1$, the right singular vector corresponding to the maximum singular value $\lambda_1$ \cite{RC98}. As a result, $\mathbf{\hat{h}}_k$ is proportional to $\lambda_1\mathbf{u}_1$, where $\mathbf{u}_1$ is the left singular vector corresponding to the maximum singular value. Hence, Assumption \ref{ass:1} is valid. It can also be shown that the assumption is valid when linear multiuser precoding is deployed.

\begin{assumption}\label{ass:2}
The entries of $\mathbf{\hat{H}}_{(-k)}$ are independent complex Gaussian random variables.
\end{assumption}
Assumption \ref{ass:2} is in general valid.  Indeed, it is obvious that the columns of $\mathbf{\hat{H}}_{(-k)}$ are independent, as they are related to channel matrices on different links. Moreover, for a given $\mathbf{t}_j$, we have $\mathbf{\hat{h}}_j\sim\mathcal{CN}(\mathbf{0}, \alpha_{k,j}||\mathbf{t}_j||_2^2\mathbf{I})$, thus implying that $\mathbf{\hat{h}}_j$ has independent entries. Being the transmit precoding vector of the $j^{th}$ primary link,  $\mathbf{t}_j$ is typically a function of $\mathbf{H}_{j,j}$ and is independent of $\mathbf{H}_{k,j}$. This justifies the validity of Assumption \ref{ass:2}.

By Assumption \ref{ass:2}, we can write $\mathbf{\hat{H}}_{(-k)}$ as $\mathbf{\hat{H}}_{(-k)} =\mathbf{\tilde{H}} \mathbf{A}$, where
\begin{eqnarray*}
 \mathbf{A}& =& \mathrm{diag}(\sqrt{\alpha_{k,1}}\|\mathbf{t}_1\|_2, \ldots, \sqrt{\alpha_{k,k-1}}\|\mathbf{t}_{k-1}\|_2, \nonumber\\
&&~~\sqrt{\alpha_{k,k+1}}\|\mathbf{t}_{k+1}\|_2, \ldots, \sqrt{\alpha_{k,K}}\|\mathbf{t}_K\|_2), 
\end{eqnarray*}
and $\mathbf{\tilde{H}}$ is an $N_k\times (K-1)$ matrix with independent standard complex Gaussian entries. Let $\mathbf{\tilde{H}}=\mathbf{\tilde{U}}\mathbf{\tilde{\Lambda}}\mathbf{\tilde{V}}^H$ be the SVD of $\mathbf{\tilde{H}}$. It is known that $\mathbf{\tilde{U}}$ and $\mathbf{\tilde{V}}$ are isotropically distributed unitary matrices\footnote{A unitary matrix is said to be isotropically distributed if its probability density is unchanged when premultiplied by a deterministic unitary matrix.}, and that $\mathbf{\tilde{U}},\mathbf{\tilde{V}},\mathbf{\tilde{\Lambda}}$ are independent \cite{MH99}.

\begin{proposition} \label{pro:1}
$\mathbf{r}_k$ has the same distribution as a normalized complex Gaussian vector as long as the random Hermitian matrix $\mathbf{W}_k$ is a unitarily invariant matrix\footnote{A random Hermitian matrix $\mathbf{W}$ is called unitarily invariant if the joint distribution of its entries equals that of $\mathbf{G}\mathbf{W}\mathbf{G}^H$ for any unitary matrix $\mathbf{G}$ independent of $\mathbf{W}$.}
\end{proposition} 

\begin{proof}
Being a unitarily invariant matrix, $\mathbf{W}_k$ can be decomposed as $\mathbf{W}_k=\mathbf{S}\mathbf{D}\mathbf{S}^H$, where $\mathbf{S}$ is an isotropically distributed unitary matrix independent of the diagonal matrix $\mathbf{D}$ \cite{TV04}. Thus, $\mathbf{r}_k=\frac{\mathbf{S}\mathbf{D}\mathbf{S}^H\hat{\mathbf{h}}_k}{||\mathbf{S}\mathbf{D}\mathbf{S}^H\hat{\mathbf{h}}_k||_2^2}$. Since the distribution of $\hat{\mathbf{h}}_k$ is rotationally invariant and $\hat{\mathbf{h}}_k$ is independent of $\mathbf{S}$, $\mathbf{S}^H\hat{\mathbf{h}}_k$ has the same distribution as $\hat{\mathbf{h}}_k$. It follows that $\mathbf{r}_k$ has the same distribution as
$$\frac{\mathbf{S}\mathbf{D}\hat{\mathbf{h}}_k}{||\mathbf{S}\mathbf{D}\hat{\mathbf{h}}_k||_2^2} = \mathbf{S}\frac{\mathbf{D}\hat{\mathbf{h}}_k}{||\mathbf{D}\hat{\mathbf{h}}_k||_2^2}.$$
Upon conditioning on $\hat{\mathbf{h}}_k$ and $\mathbf{D}$, $\frac{\mathbf{D}\hat{\mathbf{h}}_k}{||\mathbf{D}\hat{\mathbf{h}}_k||_2^2}$ is a deterministic unit vector. Since multiplying any deterministic unit vector by an isotropically distributed unitary matrix results in an isotropically distributed unit vector \cite{TV04}, the unit vector $\mathbf{r}_k$ is isotropically distributed for the given $\hat{\mathbf{h}}_k$ and $\mathbf{D}$. Since this holds for any realizations of $\hat{\mathbf{h}}_k$ and $\mathbf{D}$, it follows that $\mathbf{r}_k$ is isotropically distributed, and therefore has the same distribution as a normalized complex Gaussian vector.
\end{proof}

\begin{corollary} \label{cor:1}
$\mathbf{r}_k$ has the same distribution as a normalized complex Gaussian Gaussian vector when a MF, ZF, or MMSE receiver is deployed.
\end{corollary}
\begin{proof}
To prove Corollary \ref{cor:1},  all we need is to show that  $\mathbf{W}_k$ is unitarily invariant for MF, ZF, and MMSE receivers. As this is trivial in the MF case, where $\mathbf{W}_k=\mathbf{I}$, we will focus on the cases with ZF and MMSE receivers in the following.

For ZF receivers, substituting  $\mathbf{\hat{H}}_{(-k)} =\mathbf{\tilde{H}} \mathbf{A}$ to \eqref{eqn:ZF}, we have $\mathbf{W}_k=\mathbf{I}-\mathbf{\tilde{U}}\mathbf{\tilde{\Lambda}}(\mathbf{\tilde{\Lambda}}^H\mathbf{\tilde{\Lambda}})^{-1}\mathbf{\tilde{\Lambda}}^H\mathbf{\tilde{U}}^H=\mathbf{\tilde{U}}(\mathbf{I}-\mathbf{\tilde{\Lambda}}(\mathbf{\tilde{\Lambda}}^H\mathbf{\tilde{\Lambda}})^{-1}\mathbf{\tilde{\Lambda}}^H)\mathbf{\tilde{U}}^H$. Since $\mathbf{\tilde{U}}$ is an isotropically distributed unitary matrix, it has the same distribution as $\mathbf{G}\mathbf{\tilde{U}}$ for any unitary matrix $\mathbf{G}$ that is independent of $\mathbf{\tilde{U}}$. Therefore, $\mathbf{W}_k$ is unitarily invariant.

In the case of MMSE, \eqref{eqn:MMSE} can be written as $\mathbf{W}_k=\mathbf{\tilde{U}}(\mathbf{T}+N_0\mathbf{I})^{-1}\mathbf{\tilde{U}}^H$, where $\mathbf{T}=\mathbf{\tilde{\Lambda}}\mathbf{\tilde{V}}^H\mathbf{A}^2\mathbf{\tilde{V}}\mathbf{\tilde{\Lambda}}$ is a Hermitian matrix. Since $\mathbf{\tilde{U}}$ is an isotropically distributed unitary matrix, $\mathbf{W}_k$ is unitarily invariant.

\end{proof}

\section{Optimal Beamforming as Homogeneous QCQP}\label{section:3}
In this section, we show that the optimal SU beamforming problems that arise in the three scenarios discussed in Section \ref{sec:2} can all be formulated as quadratically constrained quadratic programming (QCQP) problems.  To begin, let us simplify Problem (\ref{eqn:generalform}) by exploiting the properties of its optimal receive beamforming solution $\mathbf{r}_S^*$.  Observe that the variable $\mathbf{r}_S$ appears only in the objective function of Problem \eqref{eqn:generalform}.  Thus, for any given $\mathbf{t}_S$, the optimal $\mathbf{r}_S$ that maximizes $\gamma_S$ is simply an MMSE receiver \cite{JD93} given by
\begin{equation}\label{eqn:rS}
\mathbf{r}_S^*(\mathbf{t}_S)=\beta_S\mathbf{\Phi}^{-1}\mathbf{H}_{S,S}\mathbf{t}_S,
\end{equation}
where $\mathbf{\Phi}=\sum_{k=1}^K\alpha_{S,k}\mathbf{H}_{S,k}\mathbf{t}_k\mathbf{t}_k^H\mathbf{H}_{S,k}^H+N_0\mathbf{I}$, and $\beta_S=\frac{1}{||\mathbf{\Phi}^{-1}\mathbf{H}_{S,S}\mathbf{t}_S||_2}$ is a normalization factor that ensures $||\mathbf{r}_S^*(\mathbf{t}_S)||_2^2=1$.

Upon substituting $\mathbf{r}_S^*(\mathbf{t}_S)$ into $\gamma_S$ (see (\ref{eq:SINR})), we obtain $\gamma_S = \alpha_{S,S}\mathbf{t}_S^H\mathbf{H}_{S,S}^H\mathbf{\Phi}^{-1}\mathbf{H}_{S,S}\mathbf{t}_S = \mathbf{t}_S^H\mathbf{A}\mathbf{t}_S$, where $\mathbf{A}=\alpha_{S,S}\mathbf{H}_{S,S}^H\mathbf{\Phi}^{-1}\mathbf{H}_{S,S}$.  In particular, we can eliminate the variable $\mathbf{r}_S$ from Problem (\ref{eqn:generalform}) and replace the objective function with the quadratic form $\mathbf{t}_S^H\mathbf{A}\mathbf{t}_S$.

\subsection{Homogeneous QCQP Formulation in Scenario 1}\label{subsection:3-2}

Now, recall that the secondary transmitter has perfect knowledge of $\mathbf{H}_{k,S}^H \mathbf{r}_k$ in Scenario 1. With the optimal $\mathbf{r}_S$ given in \eqref{eqn:rS}, Problem \eqref{eqn:generalform} becomes
\begin{subequations}\label{eqn:Scenario1}
\begin{eqnarray}
\max_{\mathbf{t}_S} && \mathbf{t}_S^H\mathbf{A}\mathbf{t}_S \label{eqn:bf-obj} \\
\text{s.t. } && \mathbf{t}_S^H\mathbf{Q}_k^1\mathbf{t}_S \leq 1 \qquad \forall k=1,\ldots,K, \label{constraint14b}\\
&& \mathbf{t}_S^H\mathbf{t}_S \leq P_{S,max}, \label{eqn:bf-power}
\end{eqnarray}
\end{subequations}
where
$
\mathbf{Q}_k^1=\frac{\alpha_{k,S}}{\epsilon_k} \mathbf{H}_{k,S}^H \mathbf{r}_k \mathbf{r}_k^H \mathbf{H}_{k,S}
$
are Hermitian positive semidefinite matrices. Problem \eqref{eqn:Scenario1} is a homogeneous QCQP, where both the objective function and inequality constraints are quadratic without linear terms.

In subsequent subsections, we will show that similar homogeneous QCQP problems can be formulated for Scenarios 2 and 3, with $\mathbf{Q}_k^1$ in \eqref{constraint14b} replaced by some suitable matrices $\mathbf{Q}_k^2$ and $\mathbf{Q}_k^3$, respectively.

\subsection{Homogeneous QCQP Formulation in Scenario 2}\label{subsection:3-3}

In Scenario 2, the realization of $\mathbf{r}_k$ is not known to the secondary transmitter.  In order to have a more efficient utilization of the spectrum, we can exploit the distribution of $\mathbf{r}_k$ and consider the probabilistic interference constraints
\begin{equation}\label{eqn:chance2}
\Pr_{\mathbf{r}_k}\bigg\{\big| \mathbf{r}_k^H \mathbf{H}_{k,S}\mathbf{t}_S\big|^2 \leq \frac{\epsilon_k}{\alpha_{k,S}}\bigg\}\geq 1-\delta_k \qquad \forall k=1,\ldots,K.
\end{equation}
To tackle the constraints in (\ref{eqn:chance2}), we need the following lemma:
\begin{lemma}\label{lem:scn2-F}
   Let $\mathbf{r}\in\mathbb{C}^n$ be a normalized complex Gaussian vector, i.e., $\mathbf{r}=\frac{\mathbf{z}}{\|\mathbf{z}\|_2}$, where $\mathbf{z}\sim \mathcal{CN}(\mathbf{0}, \mathbf{I})$ is a standard complex Gaussian vector.  Let $\mathbf{u}\in\mathbb{C}^n$ be an arbitrary vector, and let $\zeta>0,\delta\in(0,1)$ be arbitrary scalars.  Then, we have
   $$ \Pr_{\mathbf{r}}\big\{| \mathbf{r}^H \mathbf{u}|^2 \leq \zeta\big\} \ge 1-\delta \quad\Longleftrightarrow\quad  \|\mathbf{u}\|_2^2\leq \frac{\zeta}{1-\delta^{\frac{1}{n-1}}}. $$
\end{lemma}
\noindent{\emph{Proof:}}  Since the distribution of $\mathbf{r}$ is rotationally invariant, we may assume without loss of generality that $\mathbf{u}=\big[\|\mathbf{u}\|_2,0,\ldots,0\big]^T$. Then, we have
\begin{eqnarray}
&& \Pr_{\mathbf{r}}\big\{\big| \mathbf{r}^H \mathbf{u}\big|^2 \leq \zeta\big\} = \Pr\big\{\|\mathbf{u}\|_2^2\cdot\big|[\mathbf{r}]_1\big|^2 \leq \zeta\big\}\nonumber\\
  &=& \Pr\left\{(\|\mathbf{u}\|_2^2-\zeta)\big|[\mathbf{z}]_1\big|^2 \leq \zeta \sum_{i=2}^n\big|[\mathbf{z}]_i\big|^2 \right\},\nonumber
\end{eqnarray}
where $[\cdot]_i$ denotes the $i^{th}$ entry of a vector. Note that $\sum_{i=2}^n\big|[\mathbf{z}]_i\big|^2$ has the same distribution as $\frac{1}{2}\chi_{2(n-1)}^2$, where $\chi_d^2$ is the standard real chi-square random variable with $d$ degrees of freedom.  Moreover, it is independent of $\big|[\mathbf{z}]_1\big|^2$, which has the same distribution as $\frac{1}{2}\chi_2^2$.  Hence, we have
\begin{eqnarray*}
&&\Pr_{\mathbf{r}}\left\{\big| \mathbf{r}^H \mathbf{u}\big|^2 \leq \zeta\right\} = \Pr\left\{(\|\mathbf{u}\|_2^2-\zeta)\chi_2^2 \leq \zeta \chi_{2(n-1)}^2 \right\} \\
& =& \Pr\left\{ \frac{\chi_{2(n-1)}^2/2(n-1)}{\chi_2^2/2} \geq \frac{\|\mathbf{u}\|_2^2-\zeta}{\zeta(n-1)} \right\}. 
 \end{eqnarray*}
Now, let
\begin{equation}\label{eqn:incomplete}
I_\alpha(a,b)=\sum_{j=a}^{a+b-1}\frac{(a+b-1)!}{j!(a+b-1-j)!}\alpha^j(1-\alpha)^{a+b-1-j}
\end{equation}
be the regularized incomplete beta function.  It is known that the random variable
$$ F_{2(n-1),2} = \frac{\chi_{2(n-1)}^2/2(n-1)}{\chi_2^2/2} $$
follows the so-called $F$-distribution with $(2(n-1),2)$ degrees of freedom, whose cumulative distribution function (CDF) is given by $\Pr\{F_{2(n-1),2}\leq x\}=I_{\frac{2(n-1)x}{2(n-1)x+2}}(n-1,1)$.  Upon working out the summation in \eqref{eqn:incomplete}, we obtain
$$ \Pr\{F_{2(n-1),2}\leq x\}=\left(\frac{(n-1)x}{(n-1)x+1}\right)^{n-1}. $$
Thus, we conclude that
\begin{eqnarray*}
&&\Pr_{\mathbf{r}}\left\{\big| \mathbf{r}^H \mathbf{u}\big|^2 \leq \zeta\right\}\geq 1-\delta\\
\Longleftrightarrow &&\Pr\left\{F_{2(n-1),2}\leq\frac{\|\mathbf{u}\|_2^2-\zeta}{\zeta(n-1)}\right\}\leq \delta \\
\Longleftrightarrow && \left(\frac{\|\mathbf{u}\|_2^2/\zeta-1}{\|\mathbf{u}\|_2^2/\zeta}\right)^{n-1} \leq \delta\\
\Longleftrightarrow && \|\mathbf{u}\|_2^2\leq \frac{\zeta}{1-\delta^{\frac{1}{n-1}}},
\end{eqnarray*}
which completes the proof.
$\hfill\blacksquare$

By Lemma \ref{lem:scn2-F}, we can rewrite the chance constraints (\ref{eqn:chance2}) as
$$
   \|\mathbf{H}_{k,S}\mathbf{t}_S\|_2^2\leq \frac{1}{1-\delta_k^{\frac{1}{N_k-1}}}\frac{\epsilon_k}{\alpha_{k,S}} \qquad\forall k=1,\ldots, K,
$$
which of course are equivalent to
\begin{equation} \label{eq:scn2-qc}
   \mathbf{t}_S^H\mathbf{Q}_k^2\mathbf{t}_S \leq 1 \qquad\forall k=1,\ldots,K, \end{equation}
where
$$ \mathbf{Q}_k^2=\frac{1-\delta_k^{1/(N_k-1)}}{\epsilon_k}\alpha_{k,S}\mathbf{H}_{k,S}^H\mathbf{H}_{k,S}. $$
In particular, the optimal beamforming vector $\mathbf{t}_S$ in Scenario 2 can be found by solving the QCQP problem (\ref{eqn:bf-obj}), (\ref{eq:scn2-qc}) and (\ref{eqn:bf-power}).

As mentioned earlier, when $\delta_k=0$, the chance constraints \eqref{eqn:chance2} are equivalent to the worst-case constraints \eqref{eqn:worstcase}. In this case, we have $\mathbf{Q}_k^2=\frac{\alpha_{k,S}}{\epsilon_k}\mathbf{H}_{k,S}^H\mathbf{H}_{k,S}$.

\subsection{Homogeneous QCQP Formulation and Closed-Form Solution in Scenario 3}\label{subsection:3-4}

\subsubsection{\textbf{Homogeneous QCQP}}

In Scenario 3, both $\mathbf{r}_k$ and $\mathbf{H}_{k,S}$ are unknown to the secondary transmitter. Similar to Scenario 2, we consider the following probabilistic interference constraints:
\begin{equation}\label{eqn:chance3}
   \Pr_{\mathbf{r}_k, \mathbf{H}_{k,S}} \left\{\big| \mathbf{r}_k^H \mathbf{H}_{k,S}\mathbf{t}_S\big|^2 \leq \frac{\epsilon_k}{\alpha_{k,S}} \right\} \geq 1-\delta_k \qquad\forall k=1,\ldots,K.
\end{equation}
Since $\mathbf{r}_k$ and $\mathbf{H}_{k,S}$ are independent, upon conditioning on $\mathbf{r}_k$, we see that $\mathbf{r}_k^H \mathbf{H}_{k,S}\mathbf{t}_S$ is a complex Gaussian random variable with mean 0 and variance $\|\mathbf{t}_S\|_2^2$. Note that the conditional distribution of $\mathbf{r}_k^H \mathbf{H}_{k,S}\mathbf{t}_S$ is independent of $\mathbf{r}_k$. It follows that unconditionally, we have $\mathbf{r}_k^H \mathbf{H}_{k,S}\mathbf{t}_S\sim\mathcal{CN}(0,\|\mathbf{t}_S\|_2^2)$.  In particular, the random variable $|\mathbf{r}_k^H \mathbf{H}_{k,S}\mathbf{t}_S|^2$ follows an exponential distribution with parameter $\frac{1}{\|\mathbf{t}_S\|_2^2}$.  Note that this result holds as long as $\mathbf{r}_k$ is a unit-length vector, regardless of its distribution.

Since
$$ \Pr_{\mathbf{r}_k, \mathbf{H}_{k,S}} \left\{\big| \mathbf{r}_k^H \mathbf{H}_{k,S}\mathbf{t}_S\big|^2 \leq \frac{\epsilon_k}{\alpha_{k,S}}\right\} = 1-\exp\left(-\frac{\epsilon_k}{\alpha_{k,S}\|\mathbf{t}_S\|_2^2}\right), $$
it follows that the chance constraints in \eqref{eqn:chance3} can be written as
\begin{equation} \label{eq:scn3-tnorm}
   \|\mathbf{t}_S\|_2^2\leq \frac{\epsilon_k}{\alpha_{k,S}\log\frac{1}{\delta_k}} \qquad\forall k=1,\ldots,K,
\end{equation}
or equivalently,
\begin{equation} \label{eq:scn3-qc}
   \mathbf{t}_S^H\mathbf{Q}_k^3\mathbf{t}_S \le 1 \qquad\forall k=1,\ldots,K,
\end{equation}
where $\mathbf{Q}_k^3=\frac{\alpha_{k,S}}{\epsilon_k}\log\frac{1}{\delta_k}\mathbf{I}$.  Thus, the optimal beamforming vector $\mathbf{t}_S$ in Scenario 3 can be found by solving the QCQP problem (\ref{eqn:bf-obj}), (\ref{eq:scn3-qc}) and (\ref{eqn:bf-power}).

We would like to emphasize that if the worst-case interference temperature constraints \eqref{eqn:worstcase} are to be satisfied in Scenario 3, i.e., $\delta_k=0$, then the only feasible solution is $\mathbf{t}_S=\mathbf{0}$. That is, the secondary transmitter can never transmit under the overly stringent constraints. On the other hand, by allowing a small outage probability $\delta_k$, the secondary link can transmit on the same spectrum of the primary system.

\subsubsection{\textbf{Closed-Form Solution}}

Here, we show that the optimal $\mathbf{t}_S$ in Scenario 3 can be found very efficiently by a simple eigenvalue-eigenvector computation.  Indeed, observe that the constraints (\ref{eq:scn3-tnorm}) and (\ref{eqn:bf-power}) can be combined to yield the single constraint
$$
   \|\mathbf{t}_S\|_2^2 \leq \lambda \equiv \min\left\{\frac{\epsilon_1}{\alpha_{1,S}\log\frac{1}{\delta_1}}, \ldots, \frac{\epsilon_K}{\alpha_{K,S}\log\frac{1}{\delta_K}}, P_{S,max}\right\}.
$$
Thus, the optimal beamforming problem in Scenario 3, which is given by
\begin{equation}\label{eqn:eigenvalueproblem}
   \begin{array}{c@{\quad}l}
      \displaystyle{\max_{\mathbf{t}_S}} & \mathbf{t}_S^H\mathbf{A}\mathbf{t}_S \\
\text{s.t.} & \|\mathbf{t}_S\|_2^2 \leq \lambda,
   \end{array}
\end{equation}
is reduced to the problem of finding the largest eigenvalue of $\mathbf{A}$ and its associated eigenvector.  Specifically, let $\mathbf{v}$ be the eigenvector of $\mathbf{A}$ corresponding to the largest eigenvalue.  Then, the optimal solution to \eqref{eqn:eigenvalueproblem} is simply $\mathbf{t}_S^*=\sqrt{\lambda}\mathbf{v}$.  Note that there is no need to solve any QCQP in this case.

\section{SDP Relaxation}\label{section:4}
We have shown in the last section that the optimal beamforming solution can be efficiently obtained by a simple eigenvalue-eigenvector computation in Scenario 3. However, to obtain the optimal solutions for Scenarios 1 and 2, homogeneous QCQP problems of the following form have to be solved:
\begin{subequations}\label{eqn:QCQP}
\begin{eqnarray}
\max_{\mathbf{t}_S} && \mathbf{t}_S^H\mathbf{A}\mathbf{t}_S \\
\text{s.t. } && \mathbf{t}_S^H\mathbf{Q}_k\mathbf{t}_S \leq 1 \qquad\forall k=1,\ldots,K, \label{sdp-constraint14b}\\
&&\mathbf{t}_S^H\mathbf{t}_S \leq P_{S,max}\label{sdp-constraint14c}.
\end{eqnarray}
\end{subequations}
Here, $\mathbf{Q}_k$ is equal to $\mathbf{Q}_k^1$ and $\mathbf{Q}_k^2$ in Scenarios 1 and 2, respectively.  Unfortunately, since Problem (\ref{eqn:QCQP}) involves \emph{maximizing} a convex function over an intersection of $K+1$ ellipsoids, it is NP-hard in general \cite{NRT99}.  In this section, we show how Problem (\ref{eqn:QCQP}) can be tackled using semidefinite programming (SDP) relaxation methods.

To begin, observe that $\mathbf{t}_S^H\mathbf{Q}\mathbf{t}_S = \mathrm{tr}(\mathbf{Q}\mathbf{X})$ for any matrix $\mathbf{Q}$, where $\mathbf{X}=\mathbf{t}_S\mathbf{t}_S^H$ is a rank one Hermitian positive semidefinite matrix.  Thus, by relaxing the rank constraint $\mathrm{rank}(\mathbf{X})=1$, we obtain the following SDP relaxation of Problem (\ref{eqn:QCQP}):
\begin{subequations}\label{eqn:SDP1}
\begin{eqnarray}
\max_{\mathbf{X}\succeq\mathbf{0}}&& \mathrm{tr}(\mathbf{A}\mathbf{X})\\
\text{s.t. }&&\mathrm{tr}(\mathbf{Q}_k\mathbf{X}) \leq 1 \qquad\forall k=1,\ldots,K, \label{eqn:SDP-constraint1}\\
&&\mathrm{tr}(\mathbf{X}) \leq P_{S,max} \label{eqn:SDP-constraint2}.
\end{eqnarray}
\end{subequations}
The dual of (\ref{eqn:SDP1}) is given by
\begin{subequations}\label{eqn:SDP-dual}
\begin{eqnarray}
\min_{y_1,\ldots,y_K,y_{K+1}}&& \sum_{k=1}^{K}y_k+P_{S,max}y_{K+1}, \\
\text{s.t. } && \sum_{k=1}^{K} y_k\mathbf{Q}_k+y_{K+1}\mathbf{I}-\mathbf{A}\succeq\mathbf{0}, \label{eqn:SDP-dual-constraint1}\\
&& y_k\geq0 \qquad\forall k=1,\ldots,K+1.
\end{eqnarray}
\end{subequations}
It is known that SDP problems are convex and can be solved in polynomial time using standard interior-point methods \cite{VB96}.  Moreover, if we can find a rank-one optimal solution to (\ref{eqn:SDP1}), then we can extract from it an optimal solution to the original QCQP problem (\ref{eqn:QCQP}).  In this case, there is no gap between the optimal value of (\ref{eqn:QCQP}) and that of (\ref{eqn:SDP1}), and Problems (\ref{eqn:QCQP}) and (\ref{eqn:SDP1}) are equivalent.  Of course, the SDP relaxation \eqref{eqn:SDP1} is in general not equivalent to the QCQP problem (\ref{eqn:QCQP}), as we have discarded the rank constraint $\mathrm{rank}(\mathbf{X})=1$.  In the next two sections, we will discuss how to recover a rank-one solution from an optimal solution to \eqref{eqn:SDP1}.  Before proceeding, however, let us introduce the following lemma, which will be useful for our later discussions.

\begin{lemma}\label{lem:2}
Both \eqref{eqn:SDP1} and its dual \eqref{eqn:SDP-dual} satisfy the Slater condition, i.e., they are strictly feasible.
\end{lemma}
%
We omit the proof here due to page limit.
\begin{remark}
Since both \eqref{eqn:SDP1} and \eqref{eqn:SDP-dual} are strictly feasible, a pair of primal and dual feasible solutions $(\mathbf{X}^*;(y_1^*, \ldots,y_{K+1}^*))$ to \eqref{eqn:SDP1} and \eqref{eqn:SDP-dual} is optimal if and only if the following complementary conditions hold:
\begin{subequations}\label{eqn:comp}
\begin{eqnarray}
&&\mathrm{tr}\left(\mathbf{X}^*\left(\sum_{k=1}^{K} y_k^*\mathbf{Q}_k+y_{K+1}\mathbf{I}-\mathbf{A}\right)\right)=0,\label{eqn:comp1}\\
&&y_k^*\left(\mathrm{tr}(\mathbf{Q}_k\mathbf{X^*})-1\right)=0 \qquad\forall k=1,\ldots,K, \label{eqn:comp2} \\
&&y_{K+1}^*\left(\mathrm{tr}(\mathbf{X^*})-P_{S,max}\right)=0.
\end{eqnarray}
\end{subequations}
\end{remark}

\section{Optimal Rank-One Solution when $K\leq2$}\label{sec:exact}
As it turns out, when there are no more than two primary links (i.e., when $K\le2$), there is no gap between the optimal value of the SDP relaxation \eqref{eqn:SDP1} and that of the original QCQP problem \eqref{eqn:QCQP}.  Moreover, a rank-one optimal solution to (\ref{eqn:SDP1}), and hence an optimal solution to (\ref{eqn:QCQP}), can be found in polynomial time.  Specifically, we have the following proposition, which follows directly from the results of Huang et al.~\cite{HZ07,HdMZ10}:
\begin{proposition}\label{pro:nogap}
The homogeneous QCQP problem \eqref{eqn:QCQP} can be solved exactly in polynomial time when the number of primary links $K$ is at most 2. In particular, an optimal solution to (\ref{eqn:QCQP}) can be constructed from an optimal solution of \eqref{eqn:SDP1} in polynomial time.
\end{proposition}
Proposition \ref{pro:nogap} establishes the existence of polynomial-time algorithms for constructing optimal solutions to \eqref{eqn:QCQP} when $K\leq 2$.  In subsections \ref{subsection:4-1} and \ref{subsection:4-2}, we describe two such algorithms---one for the case where $K=1$, and the other for the case where $K=2$.  Both of them are based on the following decomposition theorem of Huang et al.~\cite{HZ07,HdMZ10}:
\begin{theorem}\label{thm:decompose}
Suppose that $\mathbf{Z}$ is a Hermitian positive semidefinite matrix of rank R, and $\mathbf{A}$ and $\mathbf{B}$ are two given Hermitian matrices. Then, there is a rank-one decomposition of $\mathbf{Z}$, namely, $\mathbf{Z}=\sum_{r=1}^R \mathbf{z}_r\mathbf{z}_r^H$, such that $\mathbf{z}_r^H\mathbf{A}\mathbf{z}_r=\frac{\mathrm{tr}(\mathbf{AZ})}{R}$ and $\mathbf{z}_r^H\mathbf{B}\mathbf{z}_r=\frac{\mathrm{tr}(\mathbf{BZ})}{R}$ for all $r=1,2,\ldots,R$.  Moreover, such a decomposition can be found in polynomial time.
\end{theorem}
We refer the interested readers to \cite{HdMZ10} for the proof. To be self-contained, the algorithm for computing the decomposition guaranteed by Theorem \ref{thm:decompose}, which runs in polynomial time, is given in Algorithm \ref{alg:1}.
\begin{algorithm}
\caption{Algorithm for computing the decomposition guaranteed by Theorem \ref{thm:decompose}}\label{alg:1}
\begin{algorithmic}[1]
\REQUIRE Hermitian matrices  $\mathbf{A}$ and $\mathbf{B}$, and Hermitian positive semidefinite matrix $\mathbf{Z}$ with $R=\mathrm{rank}(\mathbf{Z})$.
\ENSURE $\mathbf{Z}=\sum_{r=1}^R \mathbf{z}_r\mathbf{z}_r^H$, a rank-one decomposition of $\mathbf{Z}$ such that $\mathbf{z}_r^H\mathbf{A}\mathbf{z}_r=\frac{\mathrm{tr}(\mathbf{AZ})}{R}$, $\mathbf{z}_r^H\mathbf{B}\mathbf{z}_r=\frac{\mathrm{tr}(\mathbf{BZ})}{R}$, for $r=1,\ldots,R$.
\STATE Compute an arbitrary rank-one decomposition $\mathbf{p}_1, \mathbf{p}_2, \ldots, \mathbf{p}_R$ such that $\mathbf{Z}=\sum_{r=1}^R \mathbf{p}_r\mathbf{p}_r^H$ using, for example, Cholesky factorization.
\STATE Let $r=1$.
\REPEAT
\IF {$(\mathbf{p}_r^H\mathbf{A}\mathbf{p}_r-\frac{\mathrm{tr}(\mathbf{AZ})}{R})(\mathbf{p}_j^H\mathbf{A}\mathbf{p}_j-\frac{\mathrm{tr}(\mathbf{AZ})}{R})\geq 0$ for all $j=r+1,\ldots,R$}
\STATE $\mathbf{q}_r:=\mathbf{p}_r$.
\ELSE
\STATE Let $l\in{r+1,\ldots,R}$ be such that $(\mathbf{p}_r^H\mathbf{A}\mathbf{p}_r-\frac{\mathrm{tr}(\mathbf{AZ})}{R})(\mathbf{p}_l^H\mathbf{A}\mathbf{p}_l-\frac{\mathrm{tr}(\mathbf{AZ})}{R})<0$.
\STATE Determine $\gamma$ such that $(\mathbf{p}_r+\gamma\mathbf{p}_l)^H\mathbf{A}(\mathbf{p}_r+\gamma\mathbf{p}_l)=\frac{\mathrm{tr}(\mathbf{AZ})}{R}(1+\gamma^2)$.
\STATE $\mathbf{q}_r:=\frac{\mathbf{p}_r+\gamma\mathbf{p}_l}{\sqrt{1+\gamma^2}}$, and set $\mathbf{p}_l:=\frac{-\gamma\mathbf{p}_r+\mathbf{p}_l}{\sqrt{1+\gamma^2}}$
\ENDIF
\IF {$r=R-1$}
\STATE $\mathbf{q}_R:=\mathbf{p}_l$.
\ENDIF
\STATE $r:=r+1$.
\UNTIL $r=R-1$ \COMMENT{Comment: We have now found a decomposition $\mathbf{Z}=\sum_{r=1}^R \mathbf{q}_r\mathbf{q}_r^H$ such that $\mathbf{q}_r^H\mathbf{A}\mathbf{q}_r=\frac{\mathrm{tr}(\mathbf{AZ})}{R}$.}
\STATE Let $r=1$.
\REPEAT
\IF {$(\mathbf{q}_r^H\mathbf{B}\mathbf{q}_r-\frac{\mathrm{tr}(\mathbf{BZ})}{R})(\mathbf{q}_j^H\mathbf{B}\mathbf{q}_j-\frac{\mathrm{tr}(\mathbf{BZ})}{R})\geq 0$ for all $j=r+1,\ldots,R$}
\STATE $\mathbf{z}_r:=\mathbf{q}_r$.
\ELSE
\STATE Let $l\in{r+1,\ldots,R}$ be such that $(\mathbf{q}_r^H\mathbf{A}\mathbf{q}_r-\frac{\mathrm{tr}(\mathbf{BZ})}{R})(\mathbf{q}_l^H\mathbf{B}\mathbf{q}_l-\frac{\mathrm{tr}(\mathbf{BZ})}{R})<0$.
\STATE Compute the arguments $\alpha_1:=\arg(\mathbf{q}_r^H\mathbf{A}\mathbf{q}_l)$ and $\alpha_2:=\arg(\mathbf{q}_r^H\mathbf{B}\mathbf{q}_l)$ and the modulus $\gamma_0=|\mathbf{q}_r^H\mathbf{B}\mathbf{q}_l|$, and determine $\gamma$ such that
$
\big(\mathbf{q}_r^H\mathbf{B}\mathbf{q}_r-\frac{\mathrm{tr}(\mathbf{BZ})}{R}\big)\gamma^2+2\gamma_0\sin(\alpha_2-\alpha_1)\gamma+\mathbf{q}_l^H\mathbf{B}\mathbf{q}_l-\frac{\mathrm{tr}(\mathbf{BZ})}{R}=0. $
\STATE Set $w=\gamma e^{i(\alpha_1+\pi/2)}$.
\STATE $\mathbf{z}_r:=\frac{w\mathbf{q}_r+\mathbf{q}_l}{\sqrt{1+\gamma^2}}$, and set $\mathbf{q}_l:=\frac{-\mathbf{q}_r+\bar{w}\mathbf{q}_l}{\sqrt{1+\gamma^2}}$
\ENDIF
\IF {$r=R-1$}
\STATE $\mathbf{z}_R:=\mathbf{q}_l$.
\ENDIF
\UNTIL $r=R-1$
\end{algorithmic}
\end{algorithm}

\subsection{Optimal Rank-One Solution when $K=1$}\label{subsection:4-1}
When $K=1$, there are two quadratic constraints in \eqref{eqn:QCQP}, namely, $\mathbf{t}_S^H\mathbf{Q}_1\mathbf{t}_S \leq 1$ and $\mathbf{t}_S^H\mathbf{t}_S \leq P_{S,max}$. Let $\mathbf{X}^*$ be an arbitrary optimal solution to \eqref{eqn:SDP1}, which can be obtained in polynomial time by standard interior-point algorithms. We now show how to construct a rank-one solution to \eqref{eqn:SDP1} from $\mathbf{X}^*$.  By Theorem \ref{thm:decompose}, we can find a rank-one decomposition $\mathbf{X^*}=\sum_{r=1}^R\mathbf{t}_r\mathbf{t}_r^H$ such that
\begin{eqnarray}\label{X-decompose}
\mathbf{t}_r^H\mathbf{Q}_1\mathbf{t}_r=\frac{\mathrm{tr}(\mathbf{Q}_1\mathbf{X^*})}{R} \quad\mbox{and}\quad \mathbf{t}_r^H\mathbf{t}_r=\frac{\mathrm{tr}(\mathbf{X^*})}{R} \nonumber\\
\forall r=1,\ldots,R,
\end{eqnarray}
where $R$ is the rank of $\mathbf{X^*}$. Now, choose an arbitrary $\mathbf{t}_r$, say $\mathbf{t}_1$, and let $\mathbf{\hat{t}}=\sqrt{R}\mathbf{t}_1$. From \eqref{X-decompose} and the fact that $\mathbf{X^*}$ is a feasible for \eqref{eqn:SDP1}, we have
\begin{eqnarray}\label{eqn:feasible1}
&&\mathbf{\hat{t}}^H\mathbf{Q}_1\mathbf{\hat{t}}=\mathrm{tr}(\mathbf{Q}_1\mathbf{\hat{t}}\mathbf{\hat{t}}^H)=\mathrm{tr}(\mathbf{Q}_1\mathbf{X^*})\leq 1 \nonumber\\
&\mbox{and}&\quad \mathbf{\hat{t}}^H\mathbf{\hat{t}}=\mathrm{tr}(\mathbf{\hat{t}}\mathbf{\hat{t}}^H)=\mathrm{tr}(\mathbf{X^*})\leq P_{S,max}.
\end{eqnarray}
Thus, we conclude from \eqref{eqn:feasible1} that $\mathbf{\hat{t}}\mathbf{\hat{t}}^H$ is a rank-one feasible solution to \eqref{eqn:SDP1}.

Now, from the complementary conditions in \eqref{eqn:comp}, we have
\begin{eqnarray*}
 &&\mathrm{tr}\left(\mathbf{X}^*\left(y_1^*\mathbf{Q}_1+y_2^*\mathbf{I}-\mathbf{A}\right)\right) \\
 &=& \sum_{r=1}^R\mathrm{tr}\left(\mathbf{t}_r\mathbf{t}_r^H\left( y_1^*\mathbf{Q}_1+y_2^*\mathbf{I}-\mathbf{A} \right)\right) = 0. 
 \end{eqnarray*}
Since $y_1^*\mathbf{Q}_1+y_2^*\mathbf{I}-\mathbf{A} \succeq \mathbf{0}$ by \eqref{eqn:SDP-dual-constraint1}, we have $\mathrm{tr}\left(\mathbf{t}_r\mathbf{t}_r^H\left(y_1^*\mathbf{Q}_1+y_2^*\mathbf{I}-\mathbf{A}\right)\right)=0$ for $r=1,\ldots,R$.  This implies that $\mathrm{tr}\left(\mathbf{\hat{t}}\mathbf{\hat{t}}^H\left(y_1^*\mathbf{Q}_1+y_2^*\mathbf{I}-\mathbf{A}\right)\right)=0$.  In a similar fashion, since $\mathrm{tr}(\mathbf{Q}_1\mathbf{\hat{t}}\mathbf{\hat{t}}^H)=\mathrm{tr}(\mathbf{Q}_1\mathbf{X^*})$ and $\mathrm{tr}(\mathbf{\hat{t}}\mathbf{\hat{t}}^H)=\mathrm{tr}(\mathbf{X^*})$, we have $y_1^*\left(\mathrm{tr}(\mathbf{Q}_1\mathbf{\hat{t}}\mathbf{\hat{t}}^H)-1\right)=0$ and $y_2^*\left(\mathrm{tr}(\mathbf{\hat{t}}\mathbf{\hat{t}}^H)-P_{S,max}\right)=0$.  This, together with the feasibility conditions in \eqref{eqn:feasible1}, leads to the conclusion that the rank-one matrix $\mathbf{\hat{t}}\mathbf{\hat{t}}^H$ is an optimal solution to \eqref{eqn:SDP1}, and $\mathbf{\hat{t}}$ is an optimal solution to \eqref{eqn:QCQP}.

Note that $\mathbf{\hat{t}}$ can be computed in polynomial time, as both $\mathbf{X}^*$ and the decomposition $\mathbf{X^*}=\sum_{r=1}^R\mathbf{t}_r\mathbf{t}_r^H$ can be computed in polynomial time.

\subsection{Optimal Rank-One Solution when $K=2$}\label{subsection:4-2}
When $K=2$, there are three quadratic constraints in \eqref{eqn:QCQP}, namely, $\mathbf{t}_S^H\mathbf{Q}_1\mathbf{t}_S \leq 1$, $\mathbf{t}_S^H\mathbf{Q}_2\mathbf{t}_S \leq 1$ and $\mathbf{t}_S^H\mathbf{t}_S \leq P_{S,max}$. In the following, we show how to construct a rank-one solution from an arbitrary optimal solution $\mathbf{X}^*$ to \eqref{eqn:SDP1} in two different cases.

\subsubsection{\textbf{At least one constraint in (\ref{eqn:SDP1}) is non-binding at optimality}}
Without loss of generality, suppose that $\mathrm{tr}(\mathbf{Q}_1\mathbf{X^*}) <1$ is the non-binding constraint, while $\mathrm{tr}(\mathbf{Q}_2\mathbf{X^*}) \leq 1$ and $\mathrm{tr}(\mathbf{X^*}) \leq P_{S,max}$ can be either binding or non-binding. Due to the complementary conditions \eqref{eqn:comp2}, we must have $y_1^*=0$.  Now, construct, via Theorem \ref{thm:decompose}, a rank-one decomposition $\mathbf{X^*}=\sum_{r=1}^R\mathbf{t}_r\mathbf{t}_r^H$ such that $\mathbf{t}_r^H\mathbf{Q}_2\mathbf{t}_r=\frac{\mathrm{tr}(\mathbf{Q}_2\mathbf{X^*})}{R}$ and $\mathbf{t}_r^H\mathbf{t}_r=\frac{\mathrm{tr}(\mathbf{X^*})}{R}$ for $r=1,\ldots,R$. Since $\mathrm{tr}(\mathbf{Q}_1\mathbf{X^*})=\sum_{r=1}^R\mathbf{t}_r^H\mathbf{Q}_1\mathbf{t}_r$, there must exist an $r\in\{1,\ldots,R\}$ such that $\mathbf{t}_r^H\mathbf{Q}_1\mathbf{t}_r\leq \frac{\mathrm{tr}(\mathbf{Q}_1\mathbf{X^*})}{R}$. Without loss of generality, assume that $r=1$ and let $\mathbf{\hat{t}}=\sqrt{R}\mathbf{t}_1$. By the same argument as in the preceding subsection, the following feasible conditions and complementary conditions hold:
\begin{eqnarray*}
&&\mathrm{tr}(\mathbf{Q}_2\mathbf{\hat{t}}\mathbf{\hat{t}}^H) \leq 1,~~~
\mathrm{tr}(\mathbf{\hat{t}}\mathbf{\hat{t}}^H) \leq P_{S,max},\\
&&\mathrm{tr}\left(\mathbf{\hat{t}}\mathbf{\hat{t}}^H\left(y_1^*\mathbf{Q}_1+y_2^*\mathbf{Q}_2+y_3^*\mathbf{I}-\mathbf{A}\right)\right)=0,\\
&&y_2^*\left(\mathrm{tr}(\mathbf{Q}_2\mathbf{\hat{t}}\mathbf{\hat{t}}^H)-1\right)=0,~~
y_3^*\left(\mathrm{tr}(\mathbf{\hat{t}}\mathbf{\hat{t}}^H)-P_{S,max}\right)=0.
\end{eqnarray*}
Moreover, we have $\mathrm{tr}(\mathbf{Q}_1\mathbf{\hat{t}}\mathbf{\hat{t}}^H)=R\cdot\mathrm{tr}(\mathbf{Q}_1\mathbf{t}_1\mathbf{t}_1^H) \le \mathrm{tr}(\mathbf{Q}_1\mathbf{X^*})\leq 1$ and $y_1^*\left(\mathrm{tr}(\mathbf{Q}_1\mathbf{\hat{t}}\mathbf{\hat{t}}^H)-1\right)=0$ since $y_1^*=0$.  Hence, $\mathbf{\hat{t}}\mathbf{\hat{t}}^H$ is an optimal rank-one solution to \eqref{eqn:SDP1} and $\mathbf{\hat{t}}$ is an optimal solution to \eqref{eqn:QCQP}.

\subsubsection{\textbf{All constraints in (\ref{eqn:SDP1}) are binding at optimality}}
When all constraints are binding, we have $\mathrm{tr}(\mathbf{Q}_1\mathbf{X}^*)=\mathrm{tr}(\mathbf{Q}_2\mathbf{X}^*)=1$ and $\mathrm{tr}(\mathbf{X}^*)=P_{S,max}$. Then, we have $\mathrm{tr}((\mathbf{Q}_1-\mathbf{Q}_2)\mathbf{X}^*)=\mathrm{tr}((\mathbf{Q}_2-\frac{1}{P_{S,max}}\mathbf{I})\mathbf{X}^*)=0$.  In this case, we construct, again via Theorem \ref{thm:decompose}, a rank-one decomposition $\mathbf{X^*}=\sum_{r=1}^R\mathbf{t}_r\mathbf{t}_r^H$ such that
\begin{equation} \label{eqn:equal1}
   \mathbf{t}_r^H\left(\mathbf{Q}_1-\mathbf{Q}_2\right)\mathbf{t}_r = \frac{\mathrm{tr}\left((\mathbf{Q}_1-\mathbf{Q}_2)\mathbf{X^*}\right)}{R} = 0
\end{equation}
and
\begin{equation} \label{eqn:equal2}
   \mathbf{t}_r^H\left(\mathbf{Q}_2-\frac{1}{P_{S,max}}\mathbf{I}\right)\mathbf{t}_r = \frac{\mathrm{tr}\left(\left(\mathbf{Q}_2-\frac{1}{P_{S,max}}\mathbf{I}\right)\mathbf{X^*}\right)}{R} = 0
\end{equation}
for $r=1,\ldots,R$.  Since $\mathrm{tr}(\mathbf{Q}_1\mathbf{X}^*)=1$, there must exist an $r\in\{1,\ldots,R\}$, say $r=1$, such that $\mathrm{tr}(\mathbf{Q}_1\mathbf{t_1}\mathbf{t_1}^H)=s>0$. Let $\mathbf{\hat{t}}=\frac{\mathbf{t_1}}{\sqrt{s}}$, and consequently we have $\mathrm{tr}(\mathbf{Q}_1\mathbf{\hat{t}}\mathbf{\hat{t}}^H)=1$.  This, together with \eqref{eqn:equal1} and \eqref{eqn:equal2}, leads to $\mathrm{tr}(\mathbf{Q}_2\mathbf{\hat{t}}\mathbf{\hat{t}}^H)=1$ and $\mathrm{tr}(\mathbf{\hat{t}}\mathbf{\hat{t}}^H)=P_{S,max}$. Hence, $\mathbf{\hat{t}}\mathbf{\hat{t}}^H$ is a feasible solution, and the complementary conditions in \eqref{eqn:comp2} are satisfied. Furthermore, since $\left(y_1^*\mathbf{Q}_1+y_2^*\mathbf{Q}_2+y_3^*\mathbf{I}-\mathbf{A}\right)\succeq\mathbf{0}$, we see that $\mathrm{tr}\left(\mathbf{\hat{t}}\mathbf{\hat{t}}^H\left(y_1^*\mathbf{Q}_1+y_2^*\mathbf{Q}_2+y_3^*\mathbf{I}-\mathbf{A}\right)\right)=0$.  Hence, we conclude that $\mathbf{\hat{t}}\mathbf{\hat{t}}^H$ is an optimal rank-one solution to \eqref{eqn:SDP1}, and $\mathbf{\hat{t}}$ is an optimal solution to the QCQP problem \eqref{eqn:QCQP}. Moreover, $\mathbf{\hat{t}}$ can be computed in polynomial time.

\section{Rank-One Solution when $K\geq3$}\label{sec:largeK}
When $K\ge3$, there may not exist any rank-one optimal solution to the SDP \eqref{eqn:SDP1}.  Moreover, the QCQP problem (\ref{eqn:QCQP}) is NP-hard in general, and hence it is unlikely that we can extract, in polynomial time, an optimal solution to it from an optimal solution to the SDP (\ref{eqn:SDP1}).  However, as we shall see, we can generate a provably near-optimal solution to the QCQP problem (\ref{eqn:QCQP}) from an optimal solution to the SDP (\ref{eqn:SDP1}) using a very simple randomized procedure.

To begin, let $\mathbf{X}^*$ be an optimal solution to the SDP \eqref{eqn:SDP1}.  Define $\mathbf{Q}_{K+1}=\frac{1}{P_{S,max}}\mathbf{I}$, so that constraint \eqref{eqn:bf-power} is equivalent to $\mathbf{t}_S^H\mathbf{Q}_{K+1}\mathbf{t}_S \leq 1$.  Consider the randomized procedure outlined in Algorithm \ref{alg:2} for generating a feasible solution to (\ref{eqn:QCQP}) from $\mathbf{X}^*$.  Algorithm \ref{alg:2} can be viewed as a generalization of the procedure developed by Nemirovski et al.~\cite{NRT99} for handling {\it real} homogeneous QCQP problems.  Our goal now is to show that Algorithm \ref{alg:2} indeed returns a feasible solution to (\ref{eqn:QCQP}).  In fact, we will prove in Theorem \ref{thm:random} that the solution returned by Algorithm \ref{alg:2} is not only feasible to (\ref{eqn:QCQP}), but is also likely to be a good one, in the sense that it has an objective value that is close to the optimal value of the QCQP problem (\ref{eqn:QCQP}).  We note that such a phenomenon can also be observed from our simulations, as will be explained in the next section.

\begin{algorithm}
\caption{Generate a feasible solution to (\ref{eqn:QCQP}) from an optimal solution $\mathbf{X}^*$ to (\ref{eqn:SDP1})}\label{alg:2}
\begin{algorithmic}[1]
\REQUIRE An optimal solution $\mathbf{X}^*$ to the SDP (\ref{eqn:SDP1}).
\ENSURE A feasible solution $\mathbf{t}$ to \eqref{eqn:QCQP}.

\STATE Decompose $\mathbf{X}^*=\mathbf{\Delta}^H\mathbf{\Delta}$, where $\mathbf{\Delta}\in\mathbb{C}^{M_S \times M_S}$. Let $\widetilde{\mathbf{A}}=\mathbf{\Delta}\mathbf{A}\mathbf{\Delta}^H$ and $\widetilde{\mathbf{Q}}_k=\mathbf{\Delta}\mathbf{Q}_k\mathbf{\Delta}^H$, where $k=1,\ldots,K+1$.  It can be shown that $\mathrm{tr}(\widetilde{\mathbf{A}})=\mathrm{tr}(\mathbf{A}\mathbf{X}^*)$ and $\mathrm{tr}(\widetilde{\mathbf{Q}}_k)=\mathrm{tr}(\mathbf{Q}_k\mathbf{X}^*)\leq 1$.

\STATE Find an unitary matrix $\mathbf{U}$ that diagonalizes $\widetilde{\mathbf{A}}$, i.e., $\widehat{\mathbf{A}}=\mathbf{U}^H\widetilde{\mathbf{A}}\mathbf{U}$ is a diagonal matrix. Set $\widehat{\mathbf{Q}}_k=\mathbf{U}^H\widetilde{\mathbf{Q}}_k\mathbf{U}$.

\STATE Let $\mathbf{\xi}$ be an $M_S\times1$ random vector whose entries are independently and uniformly distributed on the unit circle in the complex plane. In other words, we have $[\mathbf{\xi}]_i=e^{j\theta_i}$, where $\theta_i$ is uniformly distributed between 0 and $2\pi$.

\STATE Return $\mathbf{t}=\frac{1}{\sqrt{ \max_k \mathbf{\xi}^H\widehat{\mathbf{Q}}_k\mathbf{\xi}, }}\mathbf{\Delta}^H\mathbf{U}\mathbf{\xi}$ as the solution.
\end{algorithmic}
\end{algorithm}


Before we introduce and prove Theorem \ref{thm:random}, let us note the following facts:
\begin{fact}\label{fact:rank}
\cite{HZ07} There exists an optimal solution to Problem \eqref{eqn:SDP1} with rank $R\leq\sqrt{K+1}$, where $K+1$ is the number of quadratic constraints. Moreover, such an optimal solution can be found in polynomial time.
\end{fact}

\begin{fact}\label{fact:Qk}
We have $\mathrm{rank}(\widehat{\mathbf{Q}}_k)=\mathrm{rank}(\tilde{\mathbf{Q}}_k)\leq \mu:=\min\{\sqrt{K+1}, M_S\}$ for $k=1,\ldots K+1$.  In particular, we can decompose $\widehat{\mathbf{Q}}_k$ as $\widehat{\mathbf{Q}}_k=\sum_{j=1}^\mu \mathbf{f}^{kj}(\mathbf{f}^{kj})^H$ for some $\mathbf{f}^{kj}\in\mathbb{C}^{M_S}$.
\end{fact}
In order to study the quality of the solution returned by Algorithm \ref{alg:2}, we need the following lemmata:
\begin{lemma}\label{lem:A}
Let $\alpha>0$ be given.  Consider the events
$$ \mathcal{A}_{kj} = \left\{\mathbf{\xi}^H\mathbf{f}^{kj}(\mathbf{f}^{kj})^H\mathbf{\xi}>\alpha\|\mathbf{f}^{kj}\|_2^2\right\},~~
\mathcal{A} = \bigcup_{k=1}^K\bigcup_{j=1}^\mu \mathcal{A}_{kj}, $$
where $\mathbf{f}^{kj}$ is obtained from the rank-one decomposition of $\widehat{\mathbf{Q}}_k$ (see Fact \ref{fact:Qk}). Then, we have
$$
\Pr\left\{\max_{1\le k\le K+1} \mathbf{\xi}^H\widehat{\mathbf{Q}}_k\mathbf{\xi}>\alpha\right\} \le \Pr\{\mathcal{A}\}.
$$
\end{lemma}
\noindent{\emph{Proof:}} If $\mathcal{A}$ does not take place, then we have $\mathbf{\xi}^H\mathbf{f}^{kj}(\mathbf{f}^{kj})^H\mathbf{\xi}\leq \alpha\|\mathbf{f}^{kj}\|_2^2$ for all $k=1,\ldots,K+1$ and $j=1,\ldots,\mu$.  This implies that
\begin{eqnarray}\label{eqn:lema1}
   \mathbf{\xi}^H\widehat{\mathbf{Q}}_k\mathbf{\xi} = \sum_{j=1}^\mu \mathbf{\xi}^H\mathbf{f}^{kj}(\mathbf{f}^{kj})^H\mathbf{\xi}\leq \alpha\sum_{j=1}^\mu\|\mathbf{f}^{kj}\|_2^2\nonumber\\
   \forall k=1,\ldots,K+1.
\end{eqnarray}
Note that $\sum_{j=1}^\mu\|\mathbf{f}^{kj}\|_2^2=\mathrm{tr}(\widehat{\mathbf{Q}}_k)=\mathrm{tr}(\tilde{\mathbf{Q}}_k)\leq1$.  Hence, \eqref{eqn:lema1} implies that $\max_{1\le k\le K+1} \mathbf{\xi}^H\widehat{\mathbf{Q}}_k\mathbf{\xi}\leq \alpha$ when $\mathcal{A}$ does not take place.  This completes the proof. $\hfill\blacksquare$
\begin{lemma} \label{lem:Hoeffding}
(Hoeffding's Inequality, Complex Version) Let $X_1,\ldots,X_n$ be independent complex-valued random variables with ${\mathbb E}X_i=0$ and $|X_i|\leq a_i$ for $i=1,\ldots,n$.  Then, for any $\beta>0$, we have
$$ \Pr\left\{\left|\sum_{i=1}^n X_i\right| > \beta\right\} \le 4\exp\left(-\frac{\beta^2}{4\|\mathbf{a}\|_2^2}\right), $$
where $\|\mathbf{a}\|_2$ denotes the $\ell_2$-norm of the vector $\mathbf{a}=[a_1, \ldots, a_n]^T$.
\end{lemma}
\noindent{\emph{Proof:}} Let $X_i^R$ and $X_i^I$ be the real and imaginary parts of $X_i$, respectively.  Then, we have ${\mathbb E}X_i^R={\mathbb E}X_i^I=0$, $|X_i^R|\le a_i$ and $|X_i^I|\le a_i$ for $i=1,\ldots,n$, and
\begin{eqnarray*}
 &&\Pr\left\{\left|\sum_{i=1}^n X_i\right| > \beta\right\}\\
  &\le& \Pr\left\{\left|\sum_{i=1}^n X_i^R\right| > \frac{\beta}{\sqrt{2}} \right\} + \Pr\left\{\left|\sum_{i=1}^n X_i^I\right| > \frac{\beta}{\sqrt{2}} \right\}. 
 \end{eqnarray*}
The desired result then follows from an application of the real version of the Hoeffding inequality \cite{H63}.
$\hfill\blacksquare$

We are now ready to present Theorem \ref{thm:random}.  It extends Nemirovski et al's result in \cite{NRT99}, which is concerned with real homogeneous QCQP problems, to the case of \emph{complex} homogeneous QCQP problems.

\begin{theorem}\label{thm:random}
The vector $\mathbf{t}$ returned by Algorithm \ref{alg:2} is well defined and is a feasible solution to Problem \eqref{eqn:QCQP}. Moreover, for any $\alpha>0$, we have
\begin{equation}\label{eqn:prob}
\Pr\left\{\mathbf{t}^H\mathbf{A}\mathbf{t}\geq \frac{1}{\alpha}\mathrm{tr}(\mathbf{A}\mathbf{X}^*)\right\}\geq 1 - 4(K+1)\mu \exp\left(-\frac{\alpha}{4}\right),
\end{equation}
where $\mu=\min\{\sqrt{K+1}, M_S\}$.  In particular, with probability at least $1-4(K+1)\mu\exp\left(-\frac{\alpha}{4}\right)$, the objective value of the solution returned by Algorithm \ref{alg:2} is at least $\frac{1}{\alpha}$ times the optimal value of the QCQP problem (\ref{eqn:QCQP}).
\end{theorem}
\noindent\emph{Proof:} We first prove that $\mathbf{t}$ is well defined, i.e., $\max_{1\le k\le K+1} \mathbf{\xi}^H\widehat{\mathbf{Q}}_k\mathbf{\xi}>0$. To see this, note that $\mathbf{\xi}^H\widehat{\mathbf{Q}}_k\mathbf{\xi}=\mathbf{\xi}^H\mathbf{U}^H\mathbf{\Delta}\mathbf{Q}_k\mathbf{\Delta}^H\mathbf{U}\mathbf{\xi}=\tilde{\mathbf{t}}^H\mathbf{Q}_k\tilde{\mathbf{t}}$, where $\tilde{\mathbf{t}}=\mathbf{\Delta}^H\mathbf{U}\mathbf{\xi}$.  Since $\mathbf{Q}_{K+1}\succ\mathbf{0}$, it follows that $\max_{1\le k\le K+1} \tilde{\mathbf{t}}^H\mathbf{Q}_k\tilde{\mathbf{t}}$, and hence $\max_{1\le k\le K+1} \mathbf{\xi}^H\widehat{\mathbf{Q}}_k\mathbf{\xi}$, must be strictly larger than zero.

Now, observe that
\begin{eqnarray*}
&&\mathbf{t}^H\mathbf{Q}_k\mathbf{t} = \frac{1}{\max_k \mathbf{\xi}^H\widehat{\mathbf{Q}}_k\mathbf{\xi}}\mathbf{\xi}^H\mathbf{U}^H\mathbf{\Delta}\mathbf{Q}_k\mathbf{\Delta}^H\mathbf{U}\mathbf{\xi} \\
&=& \frac{1}{\max_k \mathbf{\xi}^H\widehat{\mathbf{Q}}_k\mathbf{\xi}}\mathbf{\xi}^H\widehat{\mathbf{Q}}_k\mathbf{\xi} \leq 1 
\end{eqnarray*}
for $k=1,\ldots,K+1$.  It follows that $\mathbf{t}$ is a feasible solution to \eqref{eqn:QCQP}.

Next, we compute
\begin{eqnarray*}
&& \mathbf{t}^H\mathbf{A}\mathbf{t} = \frac{1}{\max_k \mathbf{\xi}^H\widehat{\mathbf{Q}}_k\mathbf{\xi}}\mathbf{\xi}^H\widehat{\mathbf{A}}\mathbf{\xi} \\
&=& \frac{1}{\max_k \mathbf{\xi}^H\widehat{\mathbf{Q}}_k\mathbf{\xi}}\mathrm{tr}(\widehat{\mathbf{A}}) = \frac{1}{\max_k \mathbf{\xi}^H\widehat{\mathbf{Q}}_k\mathbf{\xi}}\mathrm{tr}(\mathbf{A}\mathbf{X}^*), 
\end{eqnarray*}
where the second equality is due to the fact that $\widehat{\mathbf{A}}$ is a diagonal matrix and $|[\xi]_i|^2=1$ for $i=1,\ldots,M_S$.  Hence, to prove the bound in \eqref{eqn:prob}, it suffices to show that
\begin{equation} \label{eqn:prob1}
   \Pr\left\{\max_{1\le k\le K+1} \mathbf{\xi}^H\widehat{\mathbf{Q}}_k\mathbf{\xi}>\alpha\right\} < 4(K+1)\mu \exp\left(-\frac{\alpha}{4}\right).
\end{equation}
Now, by Lemma \ref{lem:A}, we have $\Pr\left\{\max_k \mathbf{\xi}^H\widehat{\mathbf{Q}}_k\mathbf{\xi}>\alpha\right\}=\Pr\{\mathcal{A}\}\leq\sum_{k,j}\Pr\{\mathcal{A}_{kj}\}$.  Moreover, since $\mathbb{E}\{[\mathbf{\xi}]_i[\mathbf{f}^{kj}]_i\}=0$ and $|[\mathbf{\xi}]_i[\mathbf{f}^{kj}]_i|=|[\mathbf{f}^{kj}]_i|$ for $i=1,\ldots,M_S$, we have
\begin{eqnarray*}
 &&\Pr\{\mathcal{A}_{kj}\} = \Pr\{|\mathbf{\xi}^H\mathbf{f}^{kj}|>\sqrt{\alpha}\|\mathbf{f}^{kj}\|_2\} \\
& =& \Pr\left\{\left|\sum_{i=1}^{M_S}[\mathbf{\xi}]_i[\mathbf{f}^{kj}]_i\right|>\sqrt{\alpha}\|\mathbf{f}^{kj}\|_2\right\} < 4\exp\left(-\frac{\alpha}{4}\right) 
\end{eqnarray*}
by Lemma \ref{lem:Hoeffding}.  This establishes \eqref{eqn:prob1} and hence the bound in \eqref{eqn:prob}.

Finally, the last statement in the theorem follows from the observation that $\mathrm{tr}(\mathbf{A}\mathbf{X}^*)$ is an upper bound on the optimal value of the QCQP problem \eqref{eqn:QCQP}, as (\ref{eqn:SDP1}) is a relaxation of (\ref{eqn:QCQP}).  This completes the proof of Theorem \ref{thm:random}.
$\hfill\blacksquare$



Before leaving this section, we emphasize that the optimal beamforming vector $\mathbf{t}_S^*$ can always be found efficiently in Scenario 3 through a matrix eigenvalue-eigenvector computation, regardless of the number of primary links. In Scenarios 1 and 2, however, the optimal solution can be obtained in polynomial time only when $K$ is no larger than two. Otherwise, we can only find an approximate solution in polynomial time via Algorithm \ref{alg:2}. Fortunately, as we will show in the next section, the approximate solution is nearly optimal most of the time.

\section{Numerical Simulations}\label{sec:simulation}
In this section, the performance of the proposed algorithms are evaluated through simulations. Throughout this section, we assume that all stations are equipped with 4 antennae. The wireless fading channel is Rayleigh distributed, and path loss exponent equals 4. The length of the secondary link is 10 meters, and $P_{S,max}$ is chosen in such a way that the average SNR received by each antenna at the secondary receiver is 10dB, if there is no interference. We also assume that all primary users transmit at power $P_{S,max}$. Likewise, the transmit beamforming vector $\mathbf{t}_k$ of primary user $k$ are set to be the dominant right singular vector of  $\mathbf{H}_{k,k}$. Meanwhile, the primary receivers use MMSE beamforming vectors, as given in Subsection II-C. Unless otherwise stated, $\delta_k$ is set to $1\%$ for all $k$ in Scenarios 2 and 3. Each point in the figures is an average of 50000 independent simulation runs.

\subsection{$K=2$}

We first investigate a network with one secondary link and two primary links. The primary links are placed such that the distances between the secondary transmitter and the two primary receivers are 15 and 13 meters, respectively, while the distances between the secondary receiver and the two primary transmitters are 12.4 and 12.7 meters. As discussed in previous sections, the optimal beamforming solution $\mathbf{t}_S^*$ can be found in polynomial time in this case.

In Fig. \ref{fig:K2SINR}, the optimal SINR $\gamma_S^*=(\mathbf{t}_S^*)^H\mathbf{A}\mathbf{t}_S^*$ (in dB scale) is plotted against $\epsilon_k$, when $\frac{\epsilon_k}{N_0}$ varies from 0 to 10 dB for all $k$. It is not surprising to see that $\gamma_S^*$ increases with the increase of the tolerable interference $\epsilon_k$ at the primary receivers. Meanwhile, the more channel information at the secondary transmitter, the higher the SINR at the secondary receiver, especially when $\epsilon_k$ is low. Noticeably, the SINR gap between the three scenarios narrows when $\epsilon_k$ increases. This is because when the primary users can tolerate higher interference levels, the secondary user can spend less effort in eliminating interference to the primary users. Hence, the advantage of knowing $\mathbf{H}_{k,S}$ and $\mathbf{r}_{k}$ becomes less obvious.

Fig. \ref{fig:K2outage} illustrates the tradeoff between the optimal SINR $\gamma_S^*$ of the secondary link and the outage probability $\delta_k$ of the primary links. It is not surprising that in both Scenarios 2 and 3, the secondary link can achieve a higher SINR when the primary links can tolerate a higher outage probability.

\begin{figure}[!ht]
\centering
\includegraphics[width=0.5\textwidth]{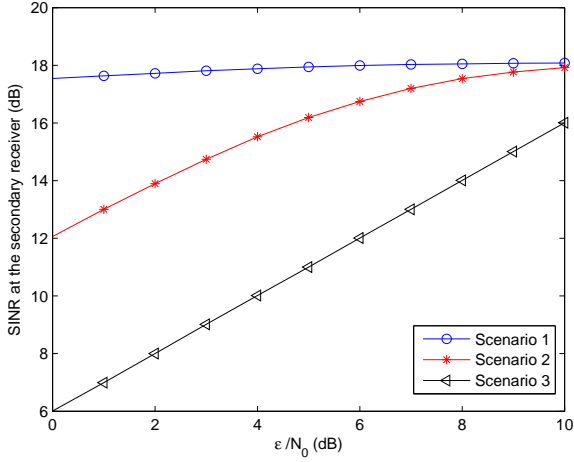}
\caption{SINR at the secondary receiver vs. $\epsilon_k/N_0$ when $K=2$.}\label{fig:K2SINR}
\end{figure}

\begin{figure}[!ht]
\centering
\includegraphics[width=0.5\textwidth]{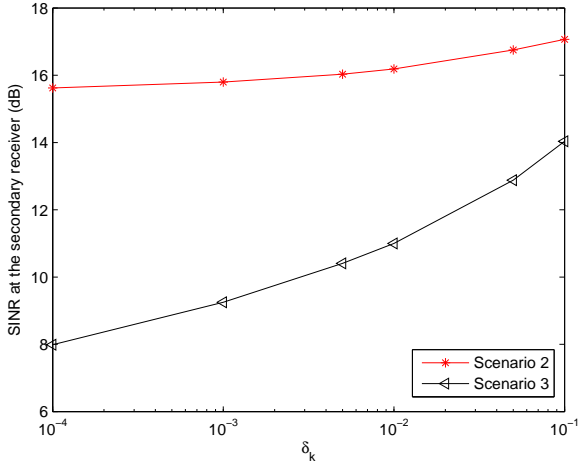}
\caption{Tradeoff between SINR and outage probability $\delta_k$ when $\epsilon_k/N_0=5dB$ and $K=2$.}\label{fig:K2outage}
\end{figure}

\subsection{$K=4$}
We now simulate a network with one secondary link and four primary links. The distance between the secondary transmitter and the four primary receivers are 20, 18, 15 and 13 meters, while that between the secondary receiver and the four primary transmitters are 16, 14, 12.4 and 13.2 meters, respectively. With four primary links, only approximate solutions can be obtained in polynomial time in Scenarios 1 and 2, as discussed in Section \ref{sec:largeK}.

In Fig. \ref{fig:K4SINR}, the randomized algorithm in Algorithm \ref{alg:2} is carried out to obtain the beamforming vector in Scenarios 1 and 2. The optimal beamforming vector in Scenario 3 is obtained through an eigenvalue-eigenvector computation. For comparison, we also plot the optimal value of the SDP relaxation \eqref{eqn:SDP1}, which is an upper bound on the maximum achievable SINR. As the figure shows, the randomized algorithm performs very close to the optimum. The achieved SINR almost overlaps with the upper bound of the optimal SINR. Meanwhile, similar conclusions drawn from Fig. \ref{fig:K2SINR} also apply here.

The tradeoff between $\gamma_S^*$ and $\delta_k$ in the four primary link case is illustrated in Fig. \ref{fig:K4outage}. Recall that in Scenario 3, the only feasible solution when $\delta_k=0$ is $\mathbf{t}_S=\mathbf{0}$. Fortunately, the achievable SINR in Scenario 3 increases rapidly with $\delta_k$ as long as $\delta_k>0$, as shown in both Fig. \ref{fig:K2SINR} and \ref{fig:K4SINR}.

\begin{figure}[!ht]
\centering
\includegraphics[width=0.5\textwidth]{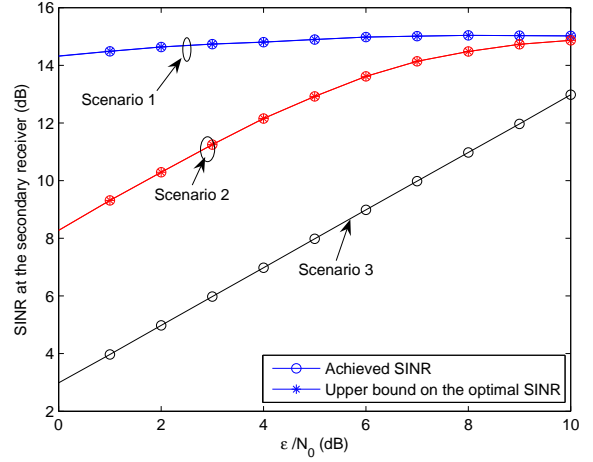}
\caption{SINR at the secondary receiver vs. $\epsilon_k/N_0$ when $K=4$.}\label{fig:K4SINR}
\end{figure}

\begin{figure}[!ht]
\centering
\includegraphics[width=0.5\textwidth]{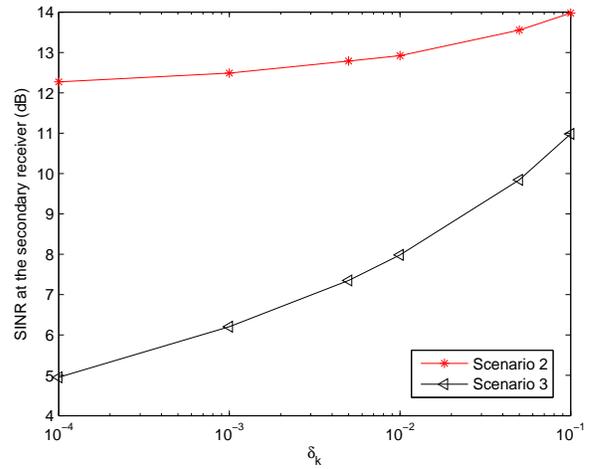}
\caption{Tradeoff between SINR and outage probability $\delta_k$ when $\epsilon_k/N_0=5dB$ and $K=4$.}\label{fig:K4outage}
\end{figure}

\subsection{A Grid Network with 9 Primary Links}
In this subsection, we consider a network with 9 primary links arranged in a 70-by-40 meter grid as shown in Fig. \ref{fig:grid}. The lengths of all link are equal to 10 meters. The secondary link is randomly placed in the area. In Fig. \ref{fig:K9SINR}, SINR at the secondary receiver is plotted against $\frac{\epsilon_k}{N_0}$. Each point in the curve is an average of 20000 independent secondary-link placements.

Again, the figure shows that Algorithm \ref{alg:2} performs very close to the optimum. The achieved SINR almost overlaps with the upper bound of its optimal value. With full CSI, Scenario 1 can achieve a much higher SINR than Scenarios 2 and 3, especially when $\epsilon_k$ is small. The better performance, however, comes at a price. In practical systems where full CSI is not available, one has to resort to the schemes developed for Scenarios 2 and 3 to achieve the maximum SINR.

\begin{figure}[!ht]
\centering
\includegraphics[width=0.5\textwidth]{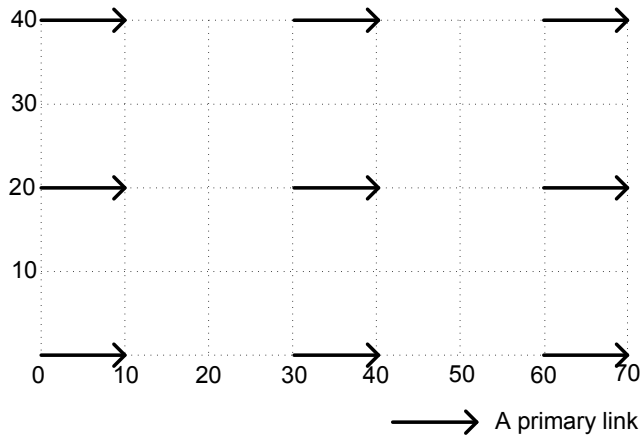}
\caption{Placement of 9 primary links.}\label{fig:grid}
\end{figure}

\begin{figure}[!ht]
\centering
\includegraphics[width=0.5\textwidth]{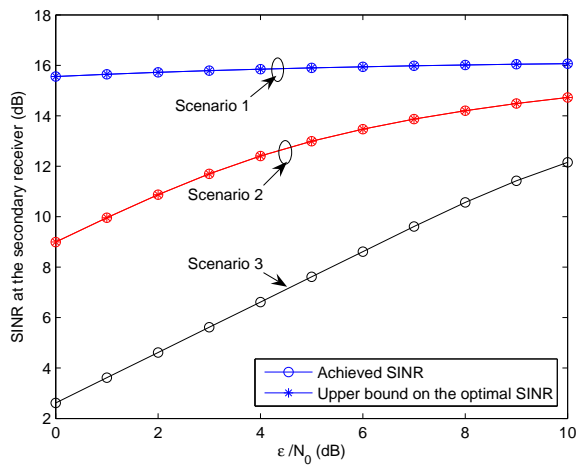}
\caption{SINR at the secondary receiver vs. $\epsilon_k/N_0$ when $K=9$.}\label{fig:K9SINR}
\end{figure}

\section{Conclusions and Discussions}\label{section:conclusions}
In this paper, we considered optimal secondary-link beamforming in MIMO CR networks when the secondary transmitter has complete, partial, or no channel knowledge on the links to primary receivers. We proposed a unified homogeneous QCQP formulation for all three scenarios with either deterministic or probabilistic interference-temperature constraints. In Scenario 3, the QCQP problem reduces to a matrix eigenvalue-eigenvector computation problem, which can be solved very efficiently. For Scenarios 1 and 2, we approached the QCQP problem by SDP relaxation. Notably, the SDP relaxation admits no gap with the true optimal value when there are no more than two primary links. In this case, the optimal beamforming solution can be computed in polynomial time. When the number of primary users exceeds two, we proposed a randomized polynomial-time algorithm that can construct a provably near-optimal solution to the QCQP problem from an optimal solution to the SDP.

The reader may notice that there is a gap between the theoretical performance of the randomized polynomial-time algorithm (Algorithm \ref{alg:2}) as established in Section \ref{sec:largeK} and its practical performance as demonstrated in Section \ref{sec:simulation}.  This can be attributed to the fact that the main theoretical result in Section \ref{sec:largeK}, namely Theorem \ref{thm:random}, only provides a {\it worst-case} guarantee on the performance of Algorithm \ref{alg:2}.  In other words, the guarantee is valid regardless of the distribution of the input data.  However, in the setting of MIMO CR networks, the input data follow a specific probability distribution, and the worst-case instance may not arise too frequently.  It would be interesting to see whether one can obtain better theoretical guarantees by performing a {\it probabilistic analysis} of the performance of Algorithm \ref{alg:2} (see \cite{SY10} and the references therein for related work).

So far, we have considered one secondary link only. However, the proposed schemes can be easily extended to a multiple-secondary-link system with the aid of medium-access-control (MAC). Suppose that there is a narrowband busy-tone channel in addition to the data-transmission channel. When a secondary link wishes to transmit a packet, it first senses the channel to see whether there is another secondary link transmitting. If not, it sends a short busy tone on the busy-tone channel to reserve the airtime. Other secondary links, upon hearing the busy tone, will keep silent during the airtime reserved by the transmitting link. Having successfully reserved the airtime, the link will then start to transmit its data packet on the data-transmission channel. In case more than one secondary transmitter sends busy tones at the same time, a collision has occurred on the busy-tone channel and the secondary transmitters will each wait for a random time period before attempting again. By doing so, it is guaranteed that there is only one secondary link transmitting data packets at a time, and the proposed optimal beamforming methods can be applied. For practical implementation, we can adopt the random-access protocols in IEEE 802.11 wireless local area networks (WLANs), such as RTS/CTS DCF, to coordinate the contention on the busy-tone channel.

Note that multiple secondary links can also coexist without the aid of a MAC protocol by properly configuring their beamforming vectors, preferably in a distributed manner. In this case, secondary links interfere with each other, and thus the optimal beamforming problem becomes much more challenging. We will address this problem in our future research.

\bibliographystyle{IEEEbib}
\bibliography{sdpbib-jsac}

\end{document}